\documentclass[journal,11pt, draftclsnofoot, onecolumn]{IEEEtran}
\usepackage{times,amsmath}
\usepackage{amsfonts}
\usepackage{amssymb,bm}
\usepackage{bbm}
\usepackage{graphicx}
\usepackage{etoolbox}
\usepackage{epsfig}
\usepackage{epstopdf}
\usepackage{algorithm}
\usepackage{algpseudocode}
\usepackage{cite}
\usepackage[caption=false]{subfig}
\usepackage{graphicx}
\usepackage{fixltx2e}
\usepackage{balance}
\usepackage{cases}
\usepackage{color}

\allowdisplaybreaks

\usepackage{amsthm}
\newtheorem{theorem}{Theorem}

\newtheorem{lemma}{Lemma}

\usepackage{chngcntr}
\counterwithout{theorem}{section}
\counterwithout{definition}{section}
\counterwithout{lemma}{section}
\counterwithout{remark}{section}
\counterwithout{assumption}{section}
\counterwithout{proposition}{section}
\counterwithout{corollary}{section}

\DeclareMathOperator{\E}{\mathbbmss{E}}

\IEEEoverridecommandlockouts
\begin{document}
%%% Consecutive citation with same author names
\bstctlcite{IEEEexample:BSTcontrol}
\setlength{\topmargin}{0in}
% paper title
\title{Centralized and Decentralized Cache-Aided Interference Management in Heterogeneous Parallel Channels}
\author{
Enrico~Piovano, Hamdi~Joudeh and Bruno~Clerckx
\thanks{The authors are with the Communications and Signal Processing group, Department of Electrical and Electronic Engineering, Imperial College London, London SW7 2AZ, U.K. (email: \{e.piovano15; hamdi.joudeh10; b.clerckx\}@imperial.ac.uk). This work is partially supported by the U.K. Engineering and Physical Sciences Research Council (EPSRC) under grant EP/N015312/1.}}
%%%
%\thanks{This work is partially supported by the U.K. Engineering and Physical Sciences Research Council (EPSRC) under grant EP/N015312/1.}
%}
%% make the title area
\maketitle
%%%%
%%
\begin{abstract}
%%%
%\vspace{-10pt}
We consider the problem of cache-aided interference management in a network consisting of $K_{\mathrm{T}}$ single-antenna transmitters and $K_{\mathrm{R}}$ single-antenna receivers, where each node is equipped with a cache memory.
%%%
Transmitters communicate with receivers over two heterogenous parallel subchannels:
the P-subchannel for which transmitters have perfect instantaneous knowledge of the  channel state, and the
N-subchannel for which the transmitters have no knowledge of the instantaneous channel state.
%%%
Under the assumptions of uncoded placement and separable one-shot linear delivery over the two subchannels,
we characterize the optimal degrees-of-freedom (DoF) to within a constant multiplicative factor of $2$.
%%%
We extend the result to a decentralized setting in which no coordination is required for content placement at the receivers.
%%%
In this case, we characterize the optimal one-shot linear DoF to within a factor of $3$.
%%%
\end{abstract}
%%%%%
%%%%%
\section{Introduction}
%%%%%%
Caching of popular content has emerged as one of the most promising strategies to cope with the
unprecedented increase in traffic over wireless and wired communication networks \cite{Golrezaei2013,Bastug2014,Liu2016,Maddah-Ali2016,Paschos2018}.
%%%
While the concept of caching is not new, its recent emergence (or re-emergence)
to the surface has been driven by a number of factors, amongst which are:
1) the nature of data network traffic which is becoming largely content-oriented due to the popularity of video-on-demand applications,
and 2) the ubiquity of memories and data storage devices.
%%%
These factors, alongside the temporal variability of data network traffic, enable nodes across the network to \emph{cache} popular
content in their cache memories during off-peak times, in which network resources are under-utilized,
and then use this cached content (sometimes in surprisingly novel ways) to alleviate the traffic load of
the network during congested peak times.
%%%

%%%
While caching has been studied within various settings and frameworks by different research communities over the past few decades \cite{Paschos2018}, recent years saw the emergence of information-theoretic studies that aim to establish the fundamental limits of cache-aided networks.
%%%
This line of research was pioneered by Maddah-Ali and Niesen in \cite{Maddah-Ali2014}, where it was shown in the context of a noiseless broadcast network that cleverly designed caching and delivery schemes yield coded-multicasting opportunities
which significantly reduce the number of required transmissions compared to conventional schemes.
%%%
This strategy, which came to be known as coded-caching, was also shown to be \emph{order-optimal} in the information-theoretic sense.
%%%
In \cite{Maddah-Ali2015a}, Maddah-Ali and Niesen further strengthened their original result by showing that the order-optimal performance of coded-caching
is attained in a decentralized alteration of the settings in \cite{Maddah-Ali2014}, where randomized content placement, requiring no central coordination amongst nodes, is employed.
%%%

%%%
This fundamental approach to caching was extended in a number of directions, including:
multi-server wired (noiseless) networks \cite{Shariatpanahi2016},
erasure and noisy broadcast networks \cite{Ghorbel2016,Amiri2018a,Bidokhti2017,Bidokhti2018,Amiri2018},
wireless device-to-device (D2D) networks \cite{Ji2016}, wireless interference networks with caches at the transmitters only or at both ends \cite{Maddah-Ali2015,Naderializadeh2017,Xu2017,Hachem2018},
multi-antenna wireless networks under a variety of assumptions regarding the availability of transmitter channel state information (CSIT) \cite{Zhang2015,Zhang2017,Lampiris2017,Piovano2017,Cao2017,Ngo2018,Shariatpanahi2017,Piovano2019,Cao2018}, and fog radio access networks (F-RANs), in which a cloud processor connects to edge nodes through front-haul links, under different assumptions and settings \cite{Sengupta2017,Kakar2017,Zhang2017a,Zhang2018,Girgis2017,Xu2018,Roig2018}.
%%%
All such works adopt information-theoretic performance
measures, i.e. capacity and its reciprocal (the latter is related to the number of transmission, or delivery time),
or their asymptotic approximations, i.e. degrees-of-freedom (DoF) and normalized delivery time (NDT).
%%%
A general observation that can be derived from these works is that caches at the transmitters enable cooperation, which is exploited through zero-forcing and interference alignment, while the redundancy arising from caches at the receivers creates coordination opportunities, exploited through coded-multicasting.
%%%%

%%%%
Most of the aforementioned works consider centralized settings, in which coordination between different nodes is
allowed during the content placement phase.
%%%
As pointed out in \cite{Maddah-Ali2015a}, while such assumption is helpful in establishing new results,
it limits their applicability as coordination may be impossible in practice, e.g. in wireless
networks where the identity and number of users is unknown beforehand.
%%%
Consequently, a number of recent works have extended the above results to decentralized scenarios including \cite{Piovano2019,Girgis2017,Xu2018,Roig2018}, which are treated in more detail after presenting this paper's setup.
%%%
%%
\subsection{The Considered Cache-Aided Wireless Network}
%%%
We consider a setup comprising a content library of $N$ files and a cache-aided wireless network consisting of $K_{\mathrm{T}}$ transmitters and
$K_{\mathrm{R}}$ receivers, each equipped with a single antenna
and a cache memory. The normalized sizes of transmitter and receiver cache memories are given by $\mu_{\mathrm{T}} \in [0,1]$ and $\mu_{\mathrm{R}}\in [0,1]$, respectively.
%%%
As commonly assumed in cache-aided systems, the network operates
in two phases: 1) a \emph{placement phase} which takes place before user demands are revealed
and in which nodes store arbitrary parts of the library according to a certain caching strategy,
and 2) a \emph{delivery phase} in which users are actively making demands for different files of
the library and in which demands are satisfied through a combination of transmissions
and the locally stored content from the placement phase.
%%%

%%%
In the considered setup, communication during the delivery phase takes place over two heterogeneous parallel subchannels: one for which
transmitters have access to the instantaneous channel coefficients (i.e. perfect CSIT),
and another for which the transmitters have no knowledge of the instantaneous channel coefficients (i.e. no CSIT).
%%%
The two subchannels are referred to as the P-subchannel and the N-subchannel, respectively.
%%%
For the sake of generality, we assume that the two subchannels occupy arbitrary fractions of the bandwidth given by $\alpha \in [0,1]$ and
$\bar{\alpha} = 1-\alpha$, respectively.
%%%
Different variants of this hybrid PN-parallel channel model have been widely adopted in information-theoretic studies focusing on capacity and DoF limits of wireless networks under CSIT imperfections (see e.g.
\cite{Tandon2013,Rassouli2016,Lashgari2016,Joudeh2019} and references therein).
%%%
This wide adoption may be attributed to the fact that the PN-parallel channel model abstracts practically relevant scenarios in which channel state feedback is available only for a fraction of signalling dimensions, e.g. sub-carriers in OFDMA systems, due to limited feedback capabilities.
%%%
Moreover, this setup and the results we obtained may also be linked to other related wireless and wired
scenarios with mixed multicast and unicast capabilities as explained further on in Section \ref{subsec:relater_setups}, 
making it all the more relevant.
%%%

%%%
In the same spirit of \cite{Naderializadeh2017}, we focus on separable one-shot linear delivery schemes
where the spreading of channel symbols over time or frequency is not allowed. This is
also known as linear precoding with no symbol extension \cite{Razaviyayn2012}.
%%%
Such linear schemes are appealing due to their practicality and their suitability for making theoretical progress on otherwise difficult or intractable information-theoretic problems.
%%%
%%
\subsection{Main Results and Contributions}
%%%
\subsubsection{Centralized Setting}
%%%
For the above described setup, we first characterize an achievable one-shot linear DoF under centralized placement
and show that it is within a factor 2 from the optimum one-shot linear DoF for all system parameters.
%%%
This achievable one-shot linear DoF is given by
%%%
\begin{equation}
\nonumber
\mathsf{DoF}_{\mathrm{L}, \mathrm{C}}(\mu_{\mathrm{T}},\mu_{\mathrm{R}},\alpha)= \alpha \cdot \min \{K_{\mathrm{T}}\mu_{\mathrm{T}}+K_{\mathrm{R}}\mu_{\mathrm{R}},K_{\mathrm{R} } \} + \bar{\alpha} \cdot  \min \{1+K_{\mathrm{R}}\mu_{\mathrm{R}},K_{\mathrm{R}}\}.
\end{equation}
%%%
From the separable nature of the proposed scheme, $\mathsf{DoF}_{\mathrm{L}, \mathrm{C}}(\mu_{\mathrm{T}},\mu_{\mathrm{R}},\alpha)$ takes a
weighted-sum form of
$\alpha \cdot \mathsf{DoF}_{\mathrm{L}, \mathrm{C}}(\mu_{\mathrm{T}},\mu_{\mathrm{R}},1) + \bar{\alpha} \cdot \mathsf{DoF}_{\mathrm{L}, \mathrm{C}}(\mu_{\mathrm{T}},\mu_{\mathrm{R}},0)$, and is hence achieved by
employing the scheme in \cite{Naderializadeh2017} over the P-subchannel and the scheme in
\cite{Maddah-Ali2014}, with a slight modification, over the N-subchannel.
%%%

%%%
To prove the order-optimality, we derive an upper bound for the one-shot linear DoF by building upon the converse proof in \cite{Naderializadeh2017},
where an integer optimization problem is formulated and then a worst-case to average demands relaxation is employed.
%%%
Further to the proof in \cite{Naderializadeh2017} however, obtaining the upper bound for the considered setup requires two more judicious steps, namely: a decoupling of the two subchannels and then a careful optimization over a delivery rate splitting ratio.
%%%
This yields an upper bound, denoted by $\mathsf{DoF}_{\mathrm{L}, \mathrm{ub}} (\mu_{\mathrm{T}},\mu_{\mathrm{R}},\alpha)$,
which also takes a weighted-sum form of
$\alpha \cdot \mathsf{DoF}_{\mathrm{L}, \mathrm{ub}}(\mu_{\mathrm{T}},\mu_{\mathrm{R}},1) + \bar{\alpha} \cdot \mathsf{DoF}_{\mathrm{L}, \mathrm{ub}}(\mu_{\mathrm{T}},\mu_{\mathrm{R}},0)$, hence reducing the task of proving order optimality to comparing
 $\mathsf{DoF}_{\mathrm{L}, \mathrm{C}}(\mu_{\mathrm{T}},\mu_{\mathrm{R}},\alpha)$ and
$\mathsf{DoF}_{\mathrm{L}, \mathrm{ub}}(\mu_{\mathrm{T}},\mu_{\mathrm{R}},\alpha)$ at the two extreme points of
$\alpha = 0$  and $\alpha = 1$ (see Sections \ref{sec:outerbound} and \ref{sec:centralized_setting}).
%%%
\subsubsection{Decentralized Setting}
%%%
The insights gained from addressing the centralized setting are then employed to tackle
a decentralized variant of the considered setup, which proves to be very technically challenging.
%%%
In the considered decentralized setting, placement at the receivers is randomized and requires no
central coordination.
%%%
On the other hand, centralized placement at the transmitters is still allowed,
as transmitters are assumed to be fixed nodes in the network, e.g.
base stations, access points or servers.
%%%
For this decentralized setting,  we show that an achievable one-shot linear DoF, which is within a
factor of 3 from the optimum one-shot linear DoF for all system parameters, is characterized by
%%%
\begin{equation}
\nonumber
\mathsf{DoF}_{\mathrm{L}, \mathrm{D}}(\mu_{\mathrm{T}},\mu_{\mathrm{R}},\alpha) =
	\alpha  \cdot \frac{1}{\sum_{l=0}^{K_{\mathrm{R}}-1} \frac{\binom{K_{\mathrm{R}}-1}{l} \mu_{\mathrm{R}}^l (1-\mu_{\mathrm{R}})^{K_{\mathrm{R}}-l-1}}{\min \{K_{\mathrm{T}}\mu_{\mathrm{T}}+l,K_{\mathrm{R}}\}}}
	  + \bar{\alpha} \cdot \frac{K_{\mathrm{R}}\mu_{\mathrm{R}}}{1-(1-\mu_{\mathrm{R}})^{K_{\mathrm{R}}}}
\end{equation}
%%%
which evidently takes the weighted-sum form of $\alpha \cdot \mathsf{DoF}_{\mathrm{L}, \mathrm{D}}(\mu_{\mathrm{T}},\mu_{\mathrm{R}},1) +
\bar{\alpha} \cdot \mathsf{DoF}_{\mathrm{L}, \mathrm{D}}(\mu_{\mathrm{T}},\mu_{\mathrm{R}},0)$.
%%

%%%
Once again, order-optimality is shown by  comparing $\mathsf{DoF}_{\mathrm{L}, \mathrm{D}}(\mu_{\mathrm{T}},\mu_{\mathrm{R}},\alpha)$  and $\mathsf{DoF}_{\mathrm{L}, \mathrm{ub}}(\mu_{\mathrm{T}},\mu_{\mathrm{R}},\alpha)$
at the two extreme points $\alpha = 0$  and $\alpha = 1$.
%%%
While the case $\alpha = 0$ follows by a direct comparison of
$\mathsf{DoF}_{\mathrm{L}, \mathrm{D}}(\mu_{\mathrm{T}},\mu_{\mathrm{R}},0)$
and $\mathsf{DoF}_{\mathrm{L}, \mathrm{ub}}(\mu_{\mathrm{T}},\mu_{\mathrm{R}},0)$,
%%%
the intricate form of $\mathsf{DoF}_{\mathrm{L}, \mathrm{D}}(\mu_{\mathrm{T}},\mu_{\mathrm{R}},1)$
does not easily lend itself to such direct approach.
%%%
Alternatively, we prove that $\frac{\mathsf{DoF}_{\mathrm{L}, \mathrm{ub}}(\mu_{\mathrm{T}},\mu_{\mathrm{R}},1)} {\mathsf{DoF}_{\mathrm{L}, \mathrm{D}}(\mu_{\mathrm{T}},\mu_{\mathrm{R}},1)} \leq
\frac{\mathsf{DoF}_{\mathrm{L}, \mathrm{ub}}(\mu_{\mathrm{T}},\mu_{\mathrm{R}},0)} {\mathsf{DoF}_{\mathrm{L}, \mathrm{D}}(\mu_{\mathrm{T}},\mu_{\mathrm{R}},0)}$, which serves the same purpose.
%%%
Showing that this last inequality holds true turns out to be particularly challenging and involves first reformulating it
as a inequality involving a polynomial, and then proving a key quasiconcavity  property for such polynomial from which the inequality follows
(see Section \ref{sec:decentralized_setting}).
\subsubsection{Related Works}
%%%
We conclude this part by highlighting the connection to other works that consider related setups.
%%%
It is evident that for $\alpha = 1$, the considered setup reduces to the one in
\cite{Naderializadeh2017,Xu2017,Hachem2018}, where only centralized placement was considered.
%%%
Since we adopt one-shot linear delivery schemes,
our work is most related to \cite{Naderializadeh2017} and expands upon it in two main directions:
1) the consideration of parallel heterogenous subchannels, and 2) the consideration of decentralized placement at the receivers.
%%%%
Another line of related works can be found in \cite{Naderializadeh2017a,Naderializadeh2019}, where 
a decentralized variant of the setting in \cite{Naderializadeh2017} was considered, with additional assumptions of 
partial connectivity and asymptotically large networks.
%%%
The latter assumption allows for a considerable simplification of the achievable DoF, which in turn, allows for a direct comparison with the 
corresponding upper bound to show order-optimality\footnote{ In particular, the achievable DoF in \cite{Naderializadeh2019} is approximated by moving a summation over the delivery time and the corresponding multicasting gains from the denominator into the numerator (see the expression of $\mathsf{DoF}_{\mathrm{L}, \mathrm{D}}(\mu_{\mathrm{T}},\mu_{\mathrm{R}},\alpha)$ for $\alpha = 1$).}.
%%%
This approach, however, does not work for the setting with finite transmitters and receivers considered here. As far as we are aware, this is the first paper that extends the results in \cite{Naderializadeh2017} to the decentralized setting without posing additional restrictions.

%%%%
The incorporation of parallel heterogeneous subchannels with the $\alpha$ parameter into cache-aided interference networks
reveals a tradeoff between CSIT feedback budget and cache sizes as observed in Section \ref{subsec:tradeoff}.
%%%%
This tradeoff extends previous observations that were made for the
cache-aided multi-antenna broadcast channel \cite{Zhang2015,Zhang2017}.
%%%%
Moreover, decentralized scenarios, which are somewhat related the setting of this work, were
considered \cite{Piovano2019,Girgis2017,Xu2018,Roig2018}.
%%%
In \cite{Piovano2019}, the multi-antenna broadcast channel with partial CSIT was considered.
%%%
While the partial CSIT setting of \cite{Piovano2019} can be translated into the parallel subchannels setting of this paper,
the full transmitter cooperation assumption (i.e. $\mu_{\mathrm{T}} = 1$) limits the
applicability of the results in \cite{Piovano2019} to the setting of this paper.
%%%
On the other hand, \cite{Girgis2017,Xu2018} consider an F-RAN setting with randomized decentralized placement
at both transmitters and receivers.
%%%
However, decentralization at both ends necessitates cloud transmission
through the front-haul in \cite{Girgis2017,Xu2018}, and the results are also not applicable to the setting considered in this work.
%%%
Finally, \cite{Roig2018} considers an F-RAN setting with similar placement to the one considered here, i.e. centralized at the transmitters and decentralized at the receivers. However, \cite{Roig2018} focuses on achievable schemes with no proofs of order-optimality.
%%%%%%%%%%%%%%%%%%%%%%%%%%%%%%%%%%%%%%%%%%%%%%%%%%%%%%%%%%%%%%%%%%%%%%%%%%%%%%%%%%%%%%%%%%%%%%%
%%%%
%%%%%%%%%%%%%%%%%%%%%%%%%%%%%%%%%%%%%%%%%%%%%%%%%%%%%%%%%%%%%%%%%%%%%%%%%%%%%%%%%%%%%%%%%%%%%%%
%%% Counters
\newcounter{Theorem_Counter}
\newcounter{Proposition_Counter}
\newcounter{Lemma_Counter}
\newcounter{Remark_Counter}
\newcounter{Assumption_Counter}
\newcounter{Definition_Counter}
%%%%%%%%%%%%%%%%%%%%%%%%%%%%%%%%%%%%%%%%%%%%%%%%%%%%%%%%%%%%%%%%%%%%%%%%%%%%%%%%%%%%%%%%%%%%%%
\section{Problem Setting} \label{sec:problem_setting}
%%%%%%%%%%%%%%%%%%%%%%%%%%%%%%%%%%%%%%%%%%%%%%%%%%%%%%%%%%%%%%%%%%%%%%%%%%%%%%%%%%%%%%%%%%%%%%
The considered wireless network consists of $K_{\mathrm{T}}$ transmitters, denoted by
$\{\text{Tx}_i\}_{i=1}^{K_{T}}$, and $K_{\mathrm{R}}$ receivers (or users), denoted by
$\{\text{Rx}_i\}_{i=1}^{K_{R}}$.
%%%
The wireless channel comprises two parallel subchannels: 1) the
P-subchannel for which the transmitters have perfect CSIT, and 2) the
N-subchannel for which the transmitters have no CSIT\footnote{Note that CSIR
is assumed to be perfectly available at all receivers.}.
%%%
We assume that the capacities of single links in the P-subchannel and the N-subchannel
are given by $\alpha \log P + o(\log P)$ and $\bar{\alpha} \log P + o(\log P)$ respectively,
where $\alpha \in [0,1]$ and $\bar{\alpha} \triangleq 1- \alpha$ are the corresponding
normalized single link capacities (or DoF) and $P$ is the SNR.
%%%
Note that under the normalization $0 \leq \alpha \leq 1$, the parameters
$\alpha$ and $\bar{\alpha}$ can be interpreted as the fractions of the total bandwidth for which
CSIT is perfect and not available respectively, in a DoF sense.
%%%

Communication over the two subchannels at time (or channel use) $t$ is modeled by
\begin{align}
\label{eq:channel_model_P}
Y_j^{(\mathrm{p})}(t) & = \sqrt{P^{\alpha}}\sum_{i=1}^{K_{\mathrm{T}}} {{h^{(\mathrm{p})}_{ji}}(t)
X^{\mathrm{(p)}}_i(t)} + Z^{(\mathrm{p})}_j (t) \\
%%%%
\label{eq:channel_model_N}
Y_j^{(\mathrm{n})}(t) & = \sqrt{P^{\bar{\alpha}}}\sum_{i=1}^{K_{\mathrm{T}}} {{h^{(\mathrm{n})}_{ji}}(t)
X^{\mathrm{(n)}}_i(t)} + Z^{(\mathrm{n})}_j (t)
\end{align}
%%%
%%
where for the P-subchannel and the N-subchannel respectively, $X^{\mathrm{(p)}}_i(t)$ and $X^{\mathrm{(n)}}_i(t)$
denote the signals transmitted by  $\text{Tx}_i$, $i \in [K_{\mathrm{T}}] \triangleq \{1,\ldots,K_{\mathrm{T}}\}$, while
$Y_j^{(\mathrm{p})}(t)$ and $Y_j^{(\mathrm{n})}(t)$
denote the signals received by $\text{Rx}_j$, $j \in [K_{\mathrm{R}}]$. Moreover, $ Z^{(\mathrm{p})}_j (t)$ and $ Z^{(\mathrm{n})}_j (t)$ denote the corresponding additive white Gaussian noise signals at $\text{Rx}_j$, distributed as $\mathcal{N}_{\mathbb{C}}(0,1)$.
%%%
${h^{(\mathrm{p})}_{ji}}(t)$ and
${h^{(\mathrm{n})}_{ji}}(t)$ denote the fading channel coefficients from $\text{Tx}_i$ to
$\text{Rx}_j$, drawn from continuous stationary ergodic processes such that ${h^{(\mathrm{p})}_{ji}}(t)$, $\forall i,j,t$, are perfectly known to the transmitters (perfect CSIT), while ${h^{(\mathrm{n})}_{ji}}(t)$, $\forall i,j,t$, are not known to the transmitters (no CSIT).
%%%
The transmit signals at $\text{Tx}_i$, $i \in [K_{\mathrm{T}}]$, are subject to the power constraints
$\mathbb{E}\big[ |X^{\mathrm{(p)}}_i(t)|^2 \big] \leq 1 $ and
$\mathbb{E}\big[ |X^{\mathrm{(n)}}_i(t)|^2 \big] \leq 1 $.
%%%
Note that $P$ is a nominal power (or SNR) value, borrowed from the  generalized degrees-of-freedom (GDoF) framework \cite{Etkin2008,Piovano2019},
which alongside $\alpha$ and $\bar{\alpha}$ is used to
distinguish the strengths of the two subchannels.
%%%

%%
In any communication session, each user requests an arbitrary file out of a content library of $N$ files given by $\mathcal{W} \triangleq \{\mathcal{W}_1, \ldots, \mathcal{W}_{N}\}$.
Following the same model in \cite{Naderializadeh2017},
each file $\mathcal{W}_n$ consists of $F$ packets, denoted by $\{\mathbf{w}_{n,f}\}_{f=1}^F$, where each packet is a vector of $B$ bits, i.e. $\mathbf{w}_{n,f} \in \mathbb{F}_2^{B}$.
Furthermore, each transmitter $\text{Tx}_i$, $i \in [K_{\mathrm{T}}]$,
is equipped with a cache memory $\mathcal{P}_i$ of size $M_{\mathrm{T}}F$ packets, while each
receiver $\text{Rx}_j$, $j \in [K_{\mathrm{R}}]$, is equipped with a cache memory $\mathcal{U}_j$ of size $M_{\mathrm{R}}F$ packets.
%%%
We assume that each cache memory, whether at transmitters or receivers, can be used to cache arbitrary contents from the library
before communication sessions begin.
%%%
Moreover, we assume that $K_{\mathrm{T}} M_{\mathrm{T}} \geq N$, which ensures that the entire library $\mathcal{W}$ can be cached across the
collective memory of all transmitters.
We define the \textit{normalized transmitter cache size} and the \textit{normalized receiver cache size}
as $\mu_{\mathrm{T}}=\frac{M_{\mathrm{T}}}{N}$ and $\mu_{\mathrm{R}}=\frac{M_{\mathrm{R}}}{N}$, respectively.
For the sake of convenience, we assume that $K_{\mathrm{T}}\mu_{\mathrm{T}}$ and $K_{\mathrm{R}}\mu_{\mathrm{R}}$ have integer values whenever we deal with the centralized case, while only $K_{\mathrm{T}}\mu_{\mathrm{T}}$ is assumed to be integer for the decentralized case. This is not a major restriction as schemes that correspond to the remaining values are realized through memory-sharing.
As commonly assumed in cache-aided systems, the network operates in two phases, a \textit{placement phase} and a \textit{delivery phase}, which are described in more detail next.
%%%%
\subsection{Placement Phase}
%%%%
The placement phase takes place before user demands are revealed and before communication sessions start.
%%%%
Following the assumptions in \cite{Naderializadeh2017}, placement is done at the packet level, i.e. each memory is filled with an arbitrary subset of the $NF$ packets in the library where the breaking of packets into smaller subpackets is not allowed.
%%%
Moreover, \textit{uncoded placement} is assumed \cite{Wan2016,Yu2018}, where it is not allowed to cache combinations of multiple packets
as a single packet.
%%%

%%%
Besides considering \emph{centralized placement}, in which coordination amongst nodes during the placement phase is allowed,
we also consider \emph{decentralized placement} where no coordination amongst receivers is allowed during the placement phase.
%%%
Centralized placement at the transmitters, however, is always assumed throughout this work, as transmitters
are considered to be fixed nodes in the network.
%%%%
\subsection{Delivery Phase}
\label{subsec:delivery_phase}
In this phase, each receiver $\text{Rx}_j$ reveals its request for an arbitrary file
$\mathcal{W}_{d_j}$, where $d_j \in [N]$. The tuple of all user demands is denoted by $\mathbf{d}=(d_1,\ldots,d_K)$.
As each receiver $\text{Rx}_j$ has the subset of requested packets, given by $\{\mathbf{w}_{d_j,f}\}_{f=1}^F \cap \mathcal{U}_j$,
pre-stored in its cache memory, the transmitters are required to deliver the remaining packets
given by $\{\mathbf{w}_{d_j,f}\}_{f=1}^F \setminus \mathcal{U}_j$, for all $j \in [K_{\mathrm{R}}]$.
Given the demands $\mathbf{d}$ and the receiver caching realization $\{\mathcal{U}_j\}_{j=1}^{K_{\mathrm{R}}}$,
the set of all packets to be delivered
is given by
\begin{equation}
\nonumber
\mathcal{D}\big( \mathbf{d},\{\mathcal{U}_j\}_{j=1}^{K_{\mathrm{R}}}\big) = \bigcup_{j = 1}^{K_{\mathrm{R}}} \big\{ \{\mathbf{w}_{d_j,f}\}_{f=1}^F \setminus \mathcal{U}_j\big\}.
\end{equation}
%%%

\emph{Packet Splitting and Encoding:}
%%%
Unlike the placement phase, in which the breaking of packets is not
allowed, we assume that each packet to be transmitted in the delivery phase is
split into two subpackets, as communication is carried out
over two parallel subchannels.
%%%
In particular, each packet $\mathbf{w}_{n,f}$ is split as
\begin{equation}
\nonumber
\mathbf{w}_{n,f}=\big(\mathbf{w}_{n,f}^{(\mathrm{p})}, \mathbf{w}_{n,f}^{(\mathrm{n})} \big)
\end{equation}
where $\mathbf{w}_{n,f}^{(\mathrm{p})}$ and $\mathbf{w}_{n,f}^{(\mathrm{n})}$ are referred to as the P-subpacket and the N-subpacket, respectively.
%%%
Without loss of generality, we assume that $\mathbf{w}_{n,f}^{(\mathrm{p})}$ and $\mathbf{w}_{n,f}^{(\mathrm{n})}$ consist of
the first $qB$ bits and the last $\bar{q}B$ bits of $\mathbf{w}_{n,f}$, respectively, where
the \textit{splitting ratio} $q \in [0,1]$ is a design parameter and $\bar{q} \triangleq 1-q$.
Moreover, while $q$ may depend on $\alpha$ (i.e. long-term channel parameters),
we assume that $q$ is fixed at the beginning of the delivery phase and is not
allowed to depend on the fading coefficients or the user demands.
From the above, each transmitter cache $\mathcal{P}_{i}$ is split into $\mathcal{P}_{i}^{(\mathrm{p})}$ and
$\mathcal{P}_{i}^{(\mathrm{c})}$, containing P-subpackets and N-subpackets respectively.
%%%%
Similarly, a set of packets to be delivered $\mathcal{D}$ is split into
$\mathcal{D}^{(\mathrm{p})}$ and $\mathcal{D}^{(\mathrm{c})}$.

%%%%
Each subpacket cached by the transmitters is encoded into a \emph{coded subpacket} using an independent random Gaussian code.
%%%
In particular, a coding scheme $\psi^{(\mathrm{p})}: \mathbb{F}_2^{qB} \rightarrow \mathbb{C}^{\tilde{B}^{(\mathrm{p})}}$ of rate $\alpha \log P + o( \log P)$ is used to encode P-subpackets, while a scheme
$\psi^{(\mathrm{n})}: \mathbb{F}_2^{(1-q)B} \rightarrow \mathbb{C}^{\tilde{B}^{(\mathrm{n})}} $  of rate $\bar{\alpha} \log P + o( \log P)$
is used to encode N-subpackets\footnote{Note that both the number of packets $F$ and the number of bits per packet $B$ may grown infinitely large.}.
%%%
The coded versions of the P-subpacket $\mathbf{w}_{n,f}^{(\mathrm{p})}$ and the N-subpacket $\mathbf{w}_{n,f}^{(\mathrm{n})}$,
defined as $\tilde{\mathbf{w}}_{n,f}^{(\mathrm{p})} \triangleq \psi^{(\mathrm{p})}(\mathbf{w}_{n,f}^{(\mathrm{p})})$ and
$\tilde{\mathbf{w}}_{n,f}^{(\mathrm{n})} \triangleq \psi^{(\mathrm{n})}(\mathbf{w}_{n,f}^{(\mathrm{n})})$ respectively,
are given in terms of channel symbols as
\begin{align}
\label{eq:coded_subpacket_p}
\tilde{\mathbf{w}}_{n,f}^{(\mathrm{p})} & =
\big( \tilde{W}_{n,f}^{(\mathrm{p})}(1), \ldots, \tilde{W}_{n,f}^{(\mathrm{p})}(\tilde{B}^{(\mathrm{p})})\big) \\
%%%
\label{eq:coded_subpacket_n}
\tilde{\mathbf{w}}_{n,f}^{(\mathrm{n})} & =
\big( \tilde{W}_{n,f}^{(\mathrm{n})}(1), \ldots, \tilde{W}_{n,f}^{(\mathrm{n})}(\tilde{B}^{(\mathrm{n})})\big).
\end{align}
%%%
It is clear that a coded P-subpacket carries a DoF of $\alpha$, while a coded
N-subpacket carries a DoF of $\bar{\alpha}$, which is in tune with the single link capacities of the corresponding subchannels.
%%%

%%%
\emph{Block Structure:}
%%%
Communication of coded subpackets is carried out independently over the P-subchannel and the N-subchannel.
%%%
Communication in the P-subchannel takes place over $H^{(\mathrm{p})}$ blocks, each referred to as a P-block and spanning
$\tilde{B}^{(\mathrm{p})}$ channel uses,
while communication in the
N-subchannel takes place over $H^{(\mathrm{n})}$ blocks, each referred to as a N-block and spanning
$\tilde{B}^{(\mathrm{n})}$ channel uses.
The goal in each P-block $m^{(\mathrm{p})} \in [H^{(\mathrm{p})}]$ is to
deliver a subset of P-subpackets $\mathcal{D}_{m^{(\mathrm{p})} }^{(\mathrm{p})} \subseteq \mathcal{D}^{(\mathrm{p})}$ to a subset of receivers, denoted by $\mathcal{R}^{(\mathrm{p})}_{m^{(\mathrm{p})} }$, such that one P-subpacket is intended exactly for one receiver.
%%%%
Similarly, in each N-block $m^{(\mathrm{n})} \in [H^{(\mathrm{n})}]$, the goal is to deliver the N-subpackets in
$\mathcal{D}_{m^{(\mathrm{n})} }^{(\mathrm{n})} \subseteq \mathcal{D}^{(\mathrm{n})}$ to the subset of receivers $\mathcal{R}^{(\mathrm{n})}_{m^{(\mathrm{n})} }$.
At the end of the communication, for each receiver $\text{Rx}_j$ to be able to retrieved its requested file, the sets of delivered subpackets and the content of the cache memory $\mathcal{U}_j$ should satisfy
\begin{align}
\mathcal{W}_{d_j}^{(\mathrm{p})} \triangleq \{\mathbf{w}_{d_j,f}^{(\mathrm{p})} \}_{f=1}^F & \subset \left(  \bigcup_{m^{(\mathrm{p})}=1}^{H^{(\mathrm{p})}}  \mathcal{D}_{m^{(\mathrm{p})} }^{(\mathrm{p})} \right) \cup \mathcal{U}_j^{(\mathrm{p})} \\
%%%
\mathcal{W}_{d_j}^{(\mathrm{n})} \triangleq \{\mathbf{w}_{d_j,f}^{(\mathrm{n})} \}_{f=1}^F & \subset \left(  \bigcup_{{m^{(\mathrm{n})}}=1}^{H^{(\mathrm{n})}}  \mathcal{D}_{m^{(\mathrm{n})}}^{(\mathrm{n})} \right) \cup \mathcal{U}_j^{(\mathrm{n})}
\end{align}
where $\mathcal{U}_j^{(\mathrm{p})}$ and $\mathcal{U}_j^{(\mathrm{n})}$ are the portions of $\mathcal{U}_j$ that correspond to
P-subpackets and N-subpackets respectively, i.e. the first $qB$ bits and the last $\bar{q}B$ bits, respectively, of packets in $\mathcal{U}_j$.
%%%%
Similarly, $\mathcal{W}_{d_j}^{(\mathrm{p})}$ and $\mathcal{W}_{d_j}^{(\mathrm{n})}$ are
the portions of $\mathcal{W}_{d_j}$ that correspond to
P-subpackets and N-subpackets respectively.
As in \cite{Naderializadeh2017}, we adopt one-shot linear delivery schemes in each subchannel, i.e.
\emph{each encoded channel symbol is beamformed  in one channel use, where spreading over multiple channel uses is not allowed}.
%%%%

\emph{Transmit Linear Beamforming:}
%%%%
Transmission of coded subpackets in each P-block and N-block is carried out using linear beamforming.
%%%
In particular, consider the $m^{(\mathrm{p})}$-th P-block, where $m^{(\mathrm{p})} \in [H^{(\mathrm{p})}]$.
%%%
$\text{Tx}_i$, $i \in [K_{\mathrm{T}}]$, transmits a linear combination of the P-subpackets in $\mathcal{P}_i^{(\mathrm{p})}$ and $\mathcal{D}^{(\mathrm{p})}_{m^{(\mathrm{p})}}$
given by
\begin{equation}
\label{eq:X_p_BF}
X^{\mathrm{(p)}}_i(t)= \sum_{ \substack{(n,f) : \\ {\mathbf{w}^{(\mathrm{p})}_{n,f} \in \mathcal{P}_i^{(\mathrm{p})} \cap \mathcal{D}^{(\mathrm{p})}_{m^{(\mathrm{p})}}   }  } }  v^{(\mathrm{p})}_{i,n,f} (t) \cdot  \tilde{W}^{(\mathrm{p})}_{n,f}(t), \
t \in \big[(m^{(\mathrm{p})} - 1)\tilde{B}^{(\mathrm{p})} + 1 : m^{(\mathrm{p})} \tilde{B}^{(\mathrm{p})}  \big]
 \end{equation}
where $[t_{1}:t_{2}] \triangleq \{t_{1},t_{1}+1,\ldots,t_{2}\}$.
In \eqref{eq:X_p_BF}, each $v^{(\mathrm{p})}_{i,n,f} (t) $ is a complex beamforming coefficient used at time $t$ over the P-subchannel,
which is allowed to depend on the channel coefficients of the P-subchannel due to perfect CSIT (e.g. as in \cite{Naderializadeh2017}).
%%%
Similarly, for the $m^{(\mathrm{n})}$-th N-block, where $m^{(\mathrm{n})} \in [H^{(\mathrm{n})}]$,
$\text{Tx}_i$ transmits a linear combination of the P-subpackets in $\mathcal{P}_i^{(\mathrm{n})}$ and $\mathcal{D}^{(\mathrm{n})}_{m^{(\mathrm{n})}}$
given by
\begin{equation}
\label{eq:X_n_BF}
X^{\mathrm{(n)}}_i(t)= \sum_{ \substack{(n,f) : \\ {\mathbf{w}^{(\mathrm{n})}_{n,f} \in \mathcal{P}_i^{(\mathrm{n})} \cap \mathcal{D}^{(\mathrm{n})}_{m^{(\mathrm{n})}}   }  } }  v^{(\mathrm{n})}_{i,n,f} (t) \cdot  \tilde{W}^{(\mathrm{n})}_{n,f}(t), \
t \in \big[(m^{(\mathrm{n})} - 1)\tilde{B}^{(\mathrm{n})} + 1 : m^{(\mathrm{n})} \tilde{B}^{(\mathrm{n})}  \big]
 \end{equation}
where each $v^{(\mathrm{n})}_{i,n,f}(t)$ is a complex beamforming coefficient, which is not allowed to depend on the channel coefficients of the N-subchannel due to no CSIT.
%%%
Note that in \eqref{eq:X_p_BF} and \eqref{eq:X_n_BF}, we implicitly assume that $\tilde{W}^{(\mathrm{p})}_{n,f}(t) =
\tilde{W}^{(\mathrm{p})}_{n,f}(t \! \! \mod \! \tilde{B}^{(\mathrm{p})} )$, $\tilde{W}^{(\mathrm{p})}_{n,f}(0) = \tilde{W}^{(\mathrm{p})}_{n,f}(\tilde{B}^{(\mathrm{p})} )$,
%%%
$\tilde{W}^{(\mathrm{n})}_{n,f}(t) =
\tilde{W}^{(\mathrm{n})}_{n,f}(t \! \! \mod \! \tilde{B}^{(\mathrm{n})} )$ and $\tilde{W}^{(\mathrm{n})}_{n,f}(0) = \tilde{W}^{(\mathrm{n})}_{n,f}(\tilde{B}^{(\mathrm{n})} )$, to maintain consistency with \eqref{eq:coded_subpacket_p} and \eqref{eq:coded_subpacket_n}.
%%%%
Moreover, the coded subpackets and beamforming coefficients are designed such that the transmit power constraints are
respected.

%%%%
\emph{Receive Linear Combining:}
%%%%
Transmit signals pass through the channel modeled in \eqref{eq:channel_model_P} and \eqref{eq:channel_model_N}.
%%%
The signals received by  $\text{Rx}_j$, $j \in [K_{\mathrm{R}}]$, in the P-block $m^{(\mathrm{p})}$
and the N-block $m^{(\mathrm{n})}$
are given by
%%%
\begin{align}
\mathbf{y}_j^{(\mathrm{p})}(m^{(\mathrm{p})}) & =
\Big(Y_{j}^{(\mathrm{p})}(t) : t \in \big[(m^{(\mathrm{p})} - 1)\tilde{B}^{(\mathrm{p})} + 1 : m^{(\mathrm{p})} \tilde{B}^{(\mathrm{p})}  \big] \Big) \\
%%%
 \mathbf{y}_j^{(\mathrm{n})}(m^{(\mathrm{n})}) & =
\Big(Y_{j}^{(\mathrm{n})}(t) : t \in \big[(m^{(\mathrm{n})} - 1)\tilde{B}^{(\mathrm{n})} + 1 : m^{(\mathrm{n})} \tilde{B}^{(\mathrm{n})}  \big] \Big)
\end{align}
%%%
where $\big(Y(t): t \in [t_{1}:t_{2}]\big) \triangleq \big(Y(t_{1}), \ldots, Y(t_{2})\big) $.
%%%
Focusing on the P-subchannel first and following the linear scheme proposed in \cite{Naderializadeh2017}, each receiver
$\text{Rx}_j$ in $\mathcal{R}^{(\mathrm{p})}_{m^{(\mathrm{p})} }$ uses the content of its cache
to subtract the interference of the undersidered subpackets in
$\mathcal{D}^{(\mathrm{p})}_{m^{(\mathrm{p})}}$, transmitted in the P-block $m^{(\mathrm{p})}$,
$m^{(\mathrm{p})} \in [H^{(\mathrm{p})}]$.
%%%%
This is achieved through a linear combination
$\mathcal{L}^{(\mathrm{p})}_{j,m^{(\mathrm{p})}}( \mathbf{y}_j^{(\mathrm{p})} (m^{(\mathrm{p})}) , \tilde{\mathcal{U}}_j^{(\mathrm{p})})$
formed to recover  $\mathbf{w}^{(\mathrm{p})}_{d_j,f} \in \mathcal{D}^{(\mathrm{p})}_{m^{(\mathrm{p})}}$, where $\tilde{\mathcal{U}}_j^{(\mathrm{p})}$ denotes the set of coded P-subpackets
cached at $\text{Rx}_j$.
%%%%
The communication in the $m^{(\mathrm{p})}$-th P-block is successful if there exists linear combinations
at the transmitters (i.e. beamformers) and linear combinations at the receivers such that
for all $\text{Rx}_j$ in $\mathcal{R}^{(\mathrm{p})}_{m^{(\mathrm{p})} }$, we have
%%%
\begin{equation}
\label{eq:linear_comb_P}
\mathcal{L}^{(\mathrm{p})}_{j,m^{(\mathrm{p})}}( \mathbf{y}_j^{(\mathrm{p})} (m^{(\mathrm{p})}) , \tilde{U}_j^{(\mathrm{p})})=
\sqrt{P^{\alpha}} \tilde{\mathbf{w}}^{(\mathrm{p})}_{d_{j},f} + \mathbf{z}^{(\mathrm{p})}_j (m^{(\mathrm{p})})
\end{equation}
%%%
where $ \mathbf{z}^{(\mathrm{p})}_j (m^{(\mathrm{p})})$ is a sequence of $\mathcal{N}_{\mathbb{C}}(0,1)$ noise samples.
%%%
The point-to-point channel in \eqref{eq:linear_comb_P} has a capacity of $\alpha \log P + o(\log P)$,
and therefore $\tilde{\mathbf{w}}^{(\mathrm{p})}_{d_{j},f}$ is reliably communicated as $qB$ grows large.
%%%

%%%
In a similar manner, considering the N-block $m^{(\mathrm{n})}$, $m^{(\mathrm{n})} \in [H^{(\mathrm{n})}]$, each receiver
$\text{Rx}_j$ in $\mathcal{R}^{(\mathrm{n})}_{m^{(\mathrm{n})} }$ forms a  linear combination
$\mathcal{L}^{(\mathrm{n})}_{j,m^{(\mathrm{n})}}( \mathbf{y}_j^{(\mathrm{n})} (m^{(\mathrm{n})}) , \tilde{\mathcal{U}}_j^{(\mathrm{n})})$
to recover $\mathbf{w}^{(\mathrm{n})}_{d_j,f} \in \mathcal{D}^{(\mathrm{n})}_{m^{(\mathrm{n})}}$, where $\tilde{\mathcal{U}}_j^{(\mathrm{n})}$ denotes the set of coded N-subpackets cached at $\text{Rx}_j$.
%%%%
The communication in the $m^{(\mathrm{n})}$-th N-block is successful if there exists linear combinations
at the transmitters  and linear combinations at the receivers such that
%%%
\begin{equation}
\label{eq:linear_comb_N}
\mathcal{L}^{(\mathrm{n})}_{j,m^{(\mathrm{n})}}( \mathbf{y}_j^{(\mathrm{n})} (m^{(\mathrm{n})}) , \tilde{\mathcal{U}}_j^{(\mathrm{n})})=
\sqrt{P^{\bar{\alpha}}} \tilde{\mathbf{w}}^{(\mathrm{n})}_{d_{j},f} + \mathbf{z}^{(\mathrm{n})}_j (m^{(\mathrm{n})})
\end{equation}
%%%
where the point-to-point channel channel in \eqref{eq:linear_comb_N} has a capacity $\bar{\alpha} \log P + o(\log P)$,
and therefore $\tilde{\mathbf{w}}^{(\mathrm{n})}_{d_{j},f}$ is reliably communicated as $\bar{q}B$ grows large.
%%%
\subsection{Delivery Time and DoF}
%%%
We start this part by defining the unit of the delivery time, i.e. the time-slot.
%%%
One time-slot is defined as the optimal time required to communicate a single packet to a single user,
under no caching and no interference, as $P \rightarrow \infty$.
%%%
This is achieved by setting $q = \alpha$, and hence communicating $\alpha B$ bits over the P-subchannel
at rate $\alpha \log P  + o(\log P)$ bits per channel use
and $\bar{\alpha}B$ bits over the N-subchannel at rate $\bar{\alpha} \log P  + o(\log P)$ bits per channel use.
%%%%
Therefore, a time-slot is equivalent to $\frac{B}{\log P }$ uses of the channel (or time instances).
%%%%
It follows that an achievable sum-DoF can be interpreted as an achievable sum-rate, measured in packets per time-slots as $P \rightarrow \infty$.
%%%

%%%
In general, for any feasible linear delivery scheme as described in Section \ref{subsec:delivery_phase},
each P-subpacket consists of $qB$ bits and is delivered in one P-block over the point-to-point channel in
\eqref{eq:linear_comb_P} at rate $\alpha \log P  + o(\log P)$.
%%%
It follows that a P-block has a duration of $\frac{q}{\alpha}$ time-slots.
%%%
Similarly, each N-subpacket consists of $\bar{q}B$ bits and is delivered over the point-to-point channel in
\eqref{eq:linear_comb_N} at rate $\bar{\alpha} \log P  + o(\log P)$,
and hence an N-block has a duration of $\frac{\bar{q}}{\bar{\alpha}}$ time-slots.
%%%
It follows that the delivery time for a feasible scheme is given by $H = \max \Big \{ \frac{q}{\alpha} H^{(\mathrm{p})}, \frac{\bar{q}}{\bar{\alpha}} H^{(\mathrm{n})} \Big \}$ time-slots,  and the achievable sum-DoF is given by $\frac{|\mathcal{D}|}{H}$.
%%%
Therefore, for fixed caching realization $\big(\{{\mathcal{P}}_i\}_{i=1}^{K_{\mathrm{T}}},\{{\mathcal{U}}_j\}_{j=1}^{K_{\mathrm{R}}}\big)$ and splitting ratio $q$,
which are independent of user demands, the maximum achievable one-shot linear sum-DoF
(DoF for short) for the worst case demands is given by
\begin{equation} \label{eq:one_shot_DoF_fix_real}
\mathsf{DoF}_{\mathrm{L}}^{({\{{\mathcal{P}}_i\}_{i=1}^{K_{\mathrm{T}}}, \{{\mathcal{U}}_j\}_{j=1}^{K_{\mathrm{R}}}},q)}=
\inf_{\mathbf{d}}
\sup_{\substack{\ H^{(\mathrm{p})}, \ H^{(\mathrm{n})},
\\ \{\mathcal{D}^{(\mathrm{p})}_{m^{(\mathrm{p})}}\}_{m^{(\mathrm{p})} = 1}^{H^{(\mathrm{p})}},
\{\mathcal{D}^{(\mathrm{n})}_{m^{(\mathrm{n})}}\}_{m^{(\mathrm{n})} = 1}^{H^{(\mathrm{n})}} }
} \frac{\left| \mathcal{D}\big( \mathbf{d},\{\mathcal{U}_j\}_{j=1}^{K_{\mathrm{R}}}\big) \right|}{\max \Big \{ \frac{q}{\alpha} H^{(\mathrm{p})}, \frac{\bar{q}}{\bar{\alpha}} H^{(\mathrm{n})} \Big \}}.
\end{equation}
This leads to the definition of the $\textit{one-shot linear DoF}$ of the network as the maximum achievable one-shot linear DoF over all caching realizations and splitting ratios, i.e.
\begin{equation}
\label{eq:network_one_shot_linear_DoF}
\begin{aligned}
 & \mathsf{DoF}_{\mathrm{L}}^*(\mu_{\mathrm{T}},\mu_{\mathrm{R}},\alpha)= & \sup_{{\{{\mathcal{P}}_i\}_{i=1}^{K_{\mathrm{T}}}, \{{\mathcal{U}}_j\}_{j=1}^{K_{\mathrm{R}}}},q} \: \: \mathsf{DoF}_{\mathrm{L}}^{({\{{\mathcal{P}}_i\}_{i=1}^{K_{\mathrm{T}}}, \{{\mathcal{U}}_j\}_{j=1}^{K_{\mathrm{R}}}},q)} \\
&& \mathrm{s.t.} \:\: |\mathcal{P}_i| = \mu_{\mathrm{T}}NF , \: \forall i \in [K_{\mathrm{T}}]\\
&&|\mathcal{U}_j| = \mu_{\mathrm{R}}NF, \: \forall j \in [K_{\mathrm{R}}] \\
&& q \in [0,1].
\end{aligned}
\end{equation}
%%
%%%%%%%%%%%%%%%%%%%%%%%%%%%%%%%%%%%%%%%%%%%%%%%%%%%%%%%%%%%%%%%%%%%%%%%%%%%%%%%%%%%%%%%%%%%%%%
\section{Main Results}
\label{sec: main_results}
%%%%%%%%%%%%%%%%%%%%
In this sections we present the main results of the paper.
%%%
The proofs are deferred to subsequent sections and appendices.
%%%%
We start with the centralized setting and then move on to the decentralized setting.
\subsection{Centralized Setting}
\begin{theorem} \label{theorem_cen}
For the cache-aided wireless network described in Section \ref{sec:problem_setting},
assuming centralized placement, an achievable one-shot linear DoF is given by
   \begin{equation} \label{eq:achievable_DoF_theo_cen}
   \mathsf{DoF}_{\mathrm{L}, \mathrm{C}}(\mu_{\mathrm{T}},\mu_{\mathrm{R}},\alpha)=
   \alpha \cdot \min \{K_{\mathrm{T}}\mu_{\mathrm{T}}+K_{\mathrm{R}}\mu_{\mathrm{R}},K_{\mathrm{R} } \} + \bar{\alpha} \cdot  \min \{1+K_{\mathrm{R}}\mu_{\mathrm{R}},K_{\mathrm{R}}\}.
   \end{equation}
   Moreover, $ \mathsf{DoF}_{\mathrm{L}, \mathrm{C}}(\mu_{\mathrm{T}},\mu_{\mathrm{R}},\alpha)$ satisfies
     \begin{equation} \label{eq:converse_theo_1}
   \frac {\mathsf{DoF}_{\mathrm{L}, \mathrm{C}}(\mu_{\mathrm{T}},\mu_{\mathrm{R}},\alpha)}{  \mathsf{DoF}_{\mathrm{L}}^{\mathrm{*}}(\mu_{\mathrm{T}},\mu_{\mathrm{R}},\alpha) } \geq \frac{1}{2},
   \end{equation}
   where $\mathsf{DoF}_{\mathrm{L}}^{\mathrm{*}}(\mu_{\mathrm{T}},\mu_{\mathrm{R}},\alpha) $ is the one-shot linear DoF of the network as
   defined in \eqref{eq:network_one_shot_linear_DoF}.
\end{theorem}
The proof of Theorem \ref{theorem_cen} is presented in Section \ref{sec:centralized_setting} and employs
the result derived in Section \ref{sec:outerbound}.
%%%

%%%
From Theorem \ref{theorem_cen}, the result in \cite[Th. 1]{Naderializadeh2017} is recovered by setting
$\alpha=1$ (P-subchannel only).
%%%
In this case, we know from \cite{Naderializadeh2017} that perfect CSIT and caches at the transmitters
allow cooperation and $\mathsf{DoF}_{\mathrm{L}, \mathrm{C}}(\mu_{\mathrm{T}},\mu_{\mathrm{R}},1)$ scales with the aggregate memory of all transmitters and receivers.
%%%
On the other hand, when $\alpha=0$ (N-subchannel only),
all DoF benefits of transmitter-side cooperation are annihilated \cite{Piovano2019},
and the achievable one-shot linear DoF in Theorem \ref{theorem_cen}
reduces to the DoF achieved with one transmitter \cite{Maddah-Ali2014}.
%%%
In this case, the original Maddah-Ali and Niesen scheme \cite{Maddah-Ali2014} is implemented,
where the XoR takes place over the air through superposition of coded packets,
and $\mathsf{DoF}_{\mathrm{L}, \mathrm{C}}(\mu_{\mathrm{T}},\mu_{\mathrm{R}},0)$ scales with the aggregate memory of the receivers only.
For general $\alpha$, $\mathsf{DoF}_{\mathrm{L}, \mathrm{C}}(\mu_{\mathrm{T}},\mu_{\mathrm{R}},\alpha)$ takes the form
\begin{equation} \label{eq:achievable_DoF_theo_cen_lin}
\mathsf{DoF}_{\mathrm{L}, \mathrm{C}}(\mu_{\mathrm{T}},\mu_{\mathrm{R}},\alpha) =  \alpha \cdot  \mathsf{DoF}_{\mathrm{L}, \mathrm{C}}(\mu_{\mathrm{T}},\mu_{\mathrm{R}},1) + \bar{\alpha} \cdot \mathsf{DoF}_{\mathrm{L}, \mathrm{C}}(\mu_{\mathrm{T}},\mu_{\mathrm{R}},0),
\end{equation}
%%%
which is achieved by choosing an adequate splitting ratio $q$ (as a function of $\alpha$)
in order to best utilize the two subchannels.
Once $q$ is chosen, the P-subpackets and N-subpackets are then delivered over the P-subchannel and N-subchannel as for
the cases with $\alpha = 1$ and $\alpha = 0$, respectively.
\subsection{Decentralized Setting}
%%%
\begin{theorem} \label{theorem_decen}
      For the cache-aided wireless network described in Section \ref{sec:problem_setting}, under decentralized placement in which
      centrally coordinated placement is only allowed at the transmitters and not at the receivers, an achievable one-shot linear DoF is given by
	\begin{equation} \label{eq:achievable_DoF_theo_decen}
	\mathsf{DoF}_{\mathrm{L}, \mathrm{D}}(\mu_{\mathrm{T}},\mu_{\mathrm{R}},\alpha) =
	\alpha  \cdot \frac{1}{\sum_{l=0}^{K_{\mathrm{R}}-1} \frac{\binom{K_{\mathrm{R}}-1}{l} \mu_{\mathrm{R}}^l (1-\mu_{\mathrm{R}})^{K_{\mathrm{R}}-l-1}}{\min \{K_{\mathrm{T}}\mu_{\mathrm{T}}+l,K_{\mathrm{R}}\}}}
	  + \bar{\alpha} \cdot \frac{K_{\mathrm{R}}\mu_{\mathrm{R}}}{1-(1-\mu_{\mathrm{R}})^{K_{\mathrm{R}}}}.
	\end{equation}
	Moreover, $\mathsf{DoF}_{\mathrm{L}, \mathrm{D}}(\mu_{\mathrm{T}},\mu_{\mathrm{R}},\alpha)$ satisfies
	\begin{equation} \label{eq:converse_theo_decen}
   \frac {\mathsf{DoF}_{\mathrm{L}, \mathrm{D}} (\mu_{\mathrm{T}},\mu_{\mathrm{R}},\alpha)}{  \mathsf{DoF}_{\mathrm{L}}^{*}(\mu_{\mathrm{T}},\mu_{\mathrm{R}},\alpha) } \geq \frac{1}{3}.
	\end{equation}
\end{theorem}
%%%
The proof of Theorem \ref{theorem_decen} is presented in Section \ref{sec:decentralized_setting}.
%%%
Choosing $\alpha=1$ in Theorem  \ref{theorem_decen}
is equivalent to considering decentralized placement for the setting of
\cite{Naderializadeh2017}. On the other hand, $\alpha=0$ reduces the setup to the
decentralized setting in \cite{Maddah-Ali2015a} in a DoF sense (the smaller multiplicative gap is due to uncoded placement and linear delivery).
%%%
In general, similar to Theorem \ref{theorem_cen}, $\mathsf{DoF}_{\mathrm{L}, \mathrm{D}}(\mu_{\mathrm{T}},\mu_{\mathrm{R}},\alpha)$ takes the form
\begin{equation}
\label{eq:achievable_DoF_theo_decen_lin}
\mathsf{DoF}_{\mathrm{L}, \mathrm{D}} (\mu_{\mathrm{T}},\mu_{\mathrm{R}},\alpha) =  \alpha \cdot \mathsf{DoF}_{\mathrm{L}, \mathrm{D}} (\mu_{\mathrm{T}},\mu_{\mathrm{R}},1) + \bar{\alpha} \cdot \mathsf{DoF}_{\mathrm{L}, \mathrm{D}} (\mu_{\mathrm{T}},\mu_{\mathrm{R}},0).
\end{equation}
%%%%%%%%%%%%%%%%%%%%%%%%%%%%%
Moreover, one could easily conclude from Theorem \ref{theorem_cen} and Theorem  \ref{theorem_decen} that centralized placement at
the receivers can only lead to at most a factor of $3$ improvement over decentralized placement.
Furthermore, we observe through numerical simulations that 
this multiplicative factor does not exceed $1.5$.
%%%%%%%%%%%%%%%%%%%%%%%%%%%%%
%%
\subsection{Tradeoff Between Receiver Cache Size and CSIT Budget}
\label{subsec:tradeoff}
\begin{figure}
\centering
\subfloat[Centralized Setting]{\includegraphics[width=.45\linewidth]{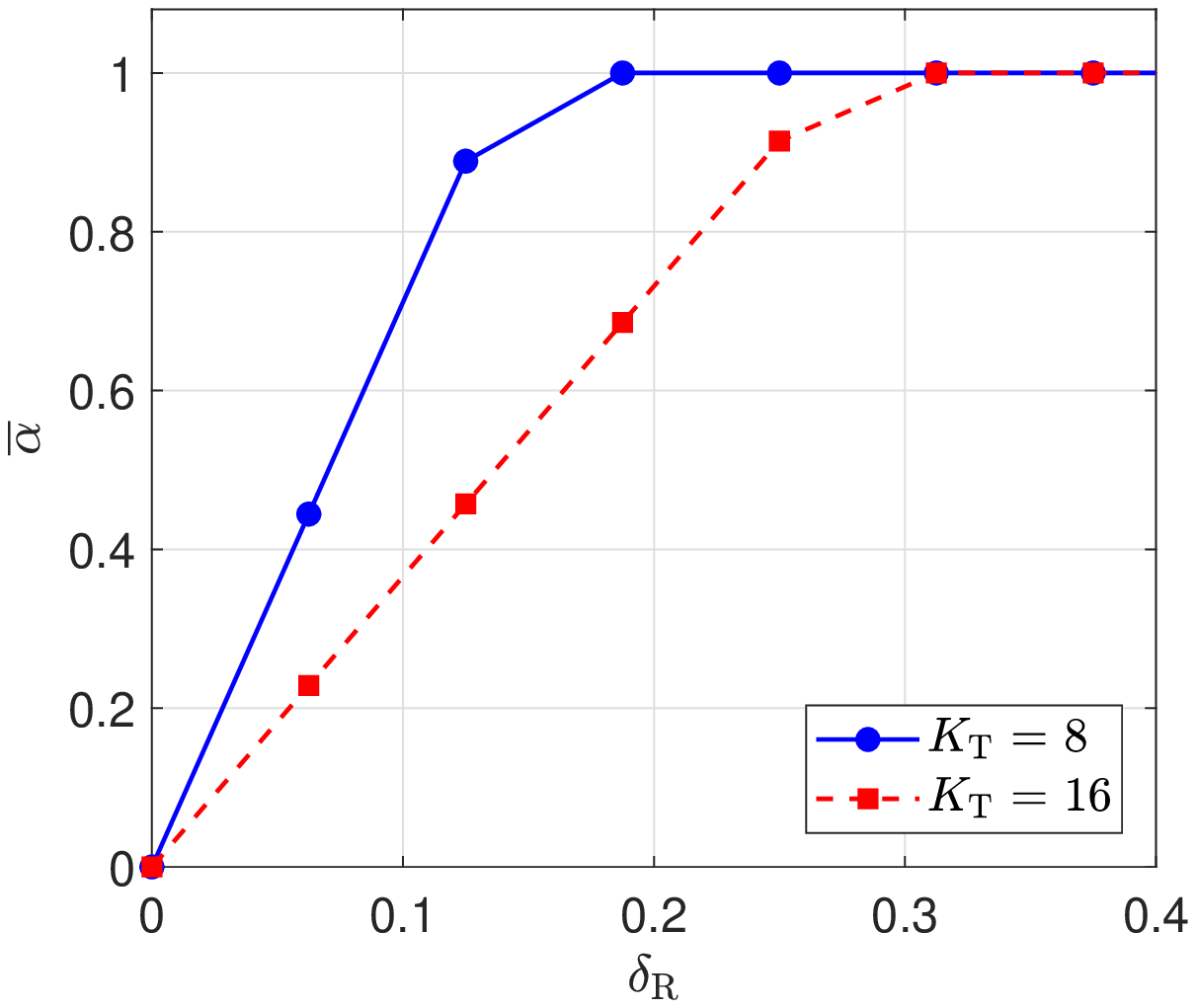}}
\subfloat[Decentralized Setting]{\includegraphics[width=.45\linewidth]{{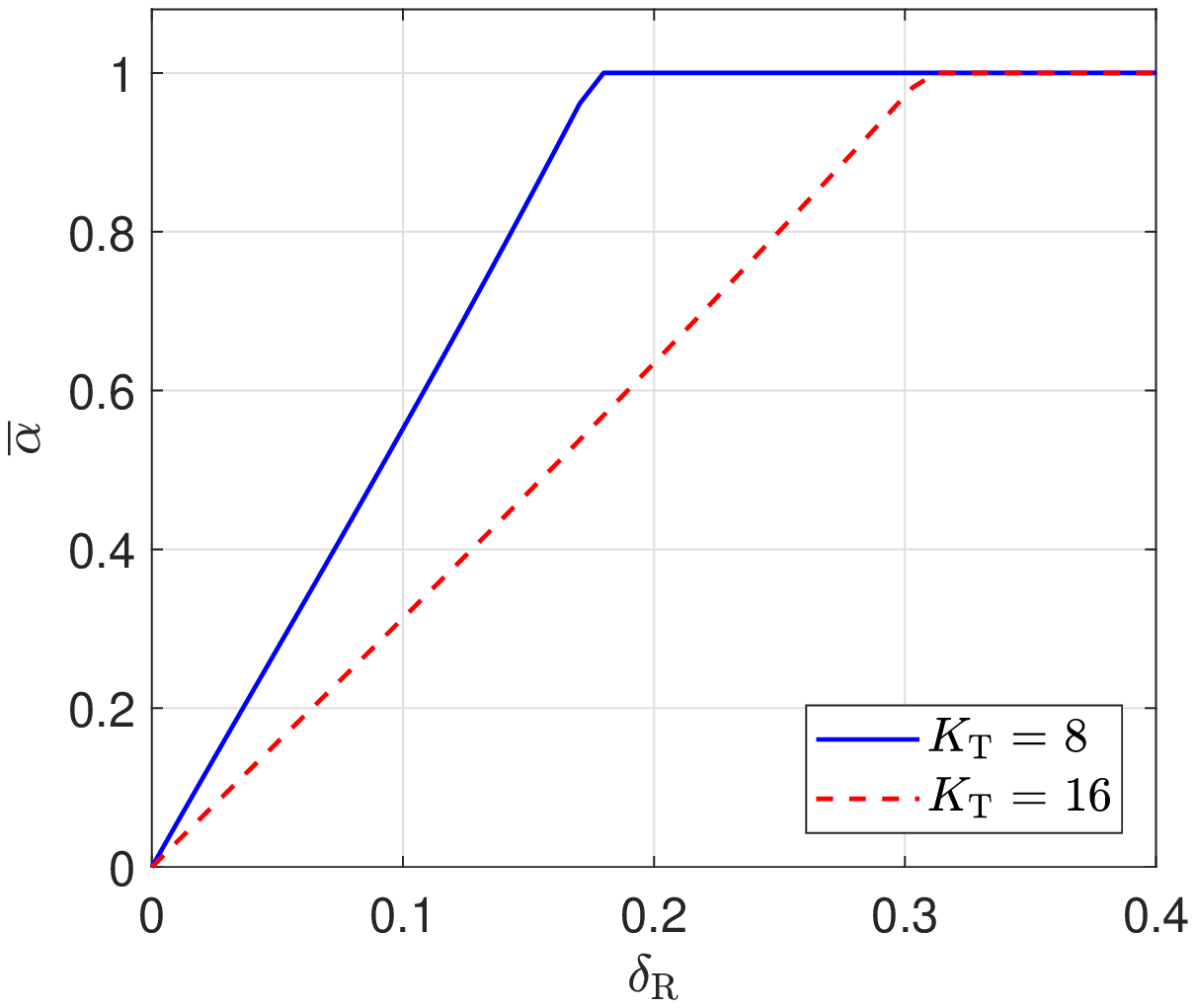}}}
\caption{Tradeoff between $\delta_{\mathrm{R}}$ and $\bar{\alpha}$ for networks with
$K_{\mathrm{R}} = 16$, $K_{\mathrm{T}} \in \{8, 16\}$, $\mu_{\mathrm{R}} = {1}/{16}$ and $\mu_{\mathrm{T}}={1}/{2}$.}
	 \label{fig:delta_alpha_tradeoff}
\end{figure}
In this part, we investigate the implications of Theorem \ref{theorem_cen} and Theorem  \ref{theorem_decen}
by considering the tradeoff between the receiver cache memory size and the CSIT budget.
%%%
For this purpose, we start by assuming that CSIT is perfectly available across all signalling dimensions, captured by $\alpha = 1$ (equivalently $\bar{\alpha} = 0$).
%%%
For given $\mu_{\mathrm{T}}$ and $\mu_{\mathrm{R}}$,
an achievable delivery time under centralized placement, denoted by
$H_{\mathrm{C}}(\mu_{\mathrm{T}},\mu_{\mathrm{R}},1)$,
is easily derived from the one-shot linear DoF in Theorem \ref{theorem_cen}.
%%%
Now suppose that the CSIT budget is reduced, e.g. by providing feedback for a fraction of sub-carriers. This yields  $H_{\mathrm{C}}(\mu_{\mathrm{T}},\mu_{\mathrm{R}},1-\bar{\alpha}) \geq H_{\mathrm{C}}(\mu_{\mathrm{T}},\mu_{\mathrm{R}},1)$, where
$\bar{\alpha}$ is interpreted as the reduction in CSIT budget.
%%%
We are interested in the corresponding increase in receiver cache size, i.e.
$\delta_{\mathrm{R}} \in [0,1-\mu_{\mathrm{R}}]$, such that $H_{\mathrm{C}}(\mu_{\mathrm{T}},\mu_{\mathrm{R}}+\delta_{\mathrm{R}},1-\bar{\alpha}) = H_{\mathrm{C}}(\mu_{\mathrm{T}},\mu_{\mathrm{R}},1)$.
%%%
Note that a similar tradeoff is defined for the decentralized case through
$H_{\mathrm{D}}(\mu_{\mathrm{T}},\mu_{\mathrm{R}}+\delta_{\mathrm{R}},1-\bar{\alpha}) = H_{\mathrm{D}}(\mu_{\mathrm{T}},\mu_{\mathrm{R}},1)$.
%%%

%%%
The tradeoff between $\mu_{\mathrm{R}}$ and $\bar{\alpha}$ is evaluated numerically and illustrated in Fig. \ref{fig:delta_alpha_tradeoff}
for both centralized and decentralized cases.
%%%
In particular, we consider a network of $K_{\mathrm{R}} = 16$ receivers with $\mu_{\mathrm{R}} = {1}/{16}$ and $\mu_{\mathrm{T}}={1}/{2}$.
%%%
The number of transmitters $K_{\mathrm{T}}$ is varied between $8$ and $16$.
%%%
It can be seen that the tradeoff is sharper for $K_{\mathrm{T}} = 8$ compared to $K_{\mathrm{T}} = 16$ in the
sense that a higher reduction in CSIT $\bar{\alpha}$ can be achieved for a smaller increase in receiver
cache size given by $\delta_{\mathrm{R}}$.
%%%
This is due to the fact that at most $8$ orthogonal beams can be created (through e.g. zero-forcing) in the setting with $K_{\mathrm{T}} = 8$, while  $K_{\mathrm{T}} = 16$ allows up to $16$ orthogonal beams.
%%%
This makes the latter setting more dependent on CSIT in general, hence requiring a higher increase in cache size to compensate
for the same reduction in CSIT budget.
%%
%%%%%%%%%%%%%%%%%%%%%%%%%%%%%
\subsection{Related Setups}
\label{subsec:relater_setups}
%%%%%%%%%%%%%%%%%%%%%%%%%%%%%
It is worthwhile highlighting that the results in Theorem \ref{theorem_cen} and Theorem \ref{theorem_decen}
can be easily applied to other related setups.
%%%
In particular, the N-subchannel can be replaced by a $(K_{\mathrm{T}}+1)$-th transmitter, operating
on a different frequency (e.g. a WiFi access point or femtocell), and connected to all transmitter caches through a
capacitated link (captured by $\bar{\alpha}$) \cite{Chen2016}.
%%%
In this case, the ergodic fading assumptions of our original setting can be relaxed, particularly if perfect CSI is also available
at the $(K_{\mathrm{T}}+1)$-th transmitter.
%%%

%%%
The results also extend to the multi-server setting of \cite{Shariatpanahi2016} with wired (noiseless) linear networks, in which
%%%
the parallel subchannels correspond to scenarios where servers can reach receivers through two parallel
networks, a fully connected linear interference network and a multicast networks.
%%%
%%%%%%%%%%%%%%%%%%%%%%%%%%%%%
%%
%%%%%%%%%%%%%%%%
\section{Centralized Setting: Proof of Theorem \ref{theorem_cen}}
\label{sec:centralized_setting}
In this section, we present a proof for Theorem \ref{theorem_cen}. As part of the proof, we introduce a DoF upper bound 
which is also used in the following section in the proof of Theorem \ref{theorem_decen}.
%%%
\subsection{Achievability of Theorem \ref{theorem_cen}}
%%%
\subsubsection{Placement Phase}
The placement phase is analogous to the one \cite{Naderializadeh2017}.
Interestingly, this implies that the placement phase is not required to depend on the value of $\alpha$.
As in \cite{Naderializadeh2017}, each file $\mathcal{W}_n$, $n \in [N]$, is partitioned into  $\binom{K_{\mathrm{T}}}{K_{\mathrm{T}}\mu_{\mathrm{T}}}\binom{K_{\mathrm{R}}}{K_{\mathrm{R}}\mu_{\mathrm{R}}}$ disjoint subfiles of equal size, denoted by
\begin{equation}
\nonumber
\mathcal{W}_n=\{\mathcal{W}_{n,\mathcal{T},\mathcal{R}}\}_{\substack{\mathcal{T} \subseteq  [K_{\mathrm{T}}] : |\mathcal{T}|= {K_{\mathrm{T}}\mu_{\mathrm{T}}} \\ \mathcal{R} \subseteq  [K_{\mathrm{R}}] : |\mathcal{R}|= {K_{\mathrm{R}}\mu_{\mathrm{R}}} } }.
\end{equation}
Note that each subfile contains $\frac{F}{\binom{K_{\mathrm{T}}}{K_{\mathrm{T}}\mu_{\mathrm{T}}}\binom{K_{\mathrm{R}}}{K_{\mathrm{R}}\mu_{\mathrm{R}}}}$ packets.
Each transmitter $\text{Tx}_i$ stores  subfiles given by $\mathcal{P}_i = \{ \mathcal{W}_{n,\mathcal{T},\mathcal{R}} : i \in \mathcal{T} \}$,
while each receiver $\text{Rx}_j$ stores subfiles given by $\mathcal{U}_j = \{ \mathcal{W}_{n,\mathcal{T},\mathcal{R}} : j \in \mathcal{R} \}$.
It is easy to verify that such placement strategy satisfies the memory size constraints at both transmitters and
receivers, and that each receiver caches $\mu_{\mathrm{R}}F$ packets from each file.
\subsubsection{Delivery Phase} \label{sec:delivery_centralized}
During the delivery phase, each receiver $\text{Rx}_j$ requests for a file $\mathcal{W}_{d_j}$.
%%%
As $\text{Rx}_j$ has all the subfiles
$\mathcal{W}_{d_j, \mathcal{T},\mathcal{R}}$ with $j \in \mathcal{R}$ cached in its memory,
%%%
it only requires the remaining subfiles given by $\mathcal{W}_{d_j, \mathcal{T},\mathcal{R}}$ with  $j \notin \mathcal{R}$.
As shown in Section \ref{subsec:delivery_phase}, each packet $\mathbf{w}_{d_j,f}$ to be delivered is split into two subpackets, i.e. $\mathbf{w}_{d_j,f} = \big(\mathbf{w}^{(\mathrm{p})}_{d_j,f}, \mathbf{w}^{(\mathrm{n})}_{d_j,n}\big)$.
%%%
We refer to the set of P-subpackets of $\mathcal{W}_{d_j, \mathcal{T},\mathcal{R}}$
as the P-subfile $\mathcal{W}_{d_j, \mathcal{T},\mathcal{R}}^{(\mathrm{p})}$, and the
set of N-packets of $\mathcal{W}_{d_j, \mathcal{T},\mathcal{R}}$
as the N-subfile $\mathcal{W}_{d_j, \mathcal{T},\mathcal{R}}^{(\mathrm{n})}$.
The P-subfiles are delivered over the P-subchannel using the linear scheme
in \cite{Naderializadeh2017}.
%%%
On the other hand, the N-subfiles are delivered over the N-subchannel using the original coded-multicasting scheme in \cite{Maddah-Ali2014},
with the difference that superposition of coded N-subpackets over the air is used instead of XoR operations before encoding,
as the latter is infeasible due to the distributed nature of transmitters.
Decoding of subpackets at the receivers is carried out after taking the appropriate linear combinations, e.g. see
(\ref{eq:linear_comb_P})  and (\ref{eq:linear_comb_N}).
Each $\text{Rx}_j$ retrieves all missing P-subfiles and N-subfile and hence the file $\mathcal{W}_{d_j}$ is recovered.
%%%
\subsubsection{Achievable One-Shot Linear DoF}
Since each user has $\mu_{\mathrm{R}}F$ packets from each file stored in its cache memory, a total of
$K_{\mathrm{R}}F(1-\mu_{\mathrm{R}}) $ packets are delivered during the delivery phase, split into
$K_{\mathrm{R}}F(1-\mu_{\mathrm{R}}) $ P-subpackets and $K_{\mathrm{R}}F(1-\mu_{\mathrm{R}})$ N-subpackets
delivered over the P-subchannel and N-subchannel, respectively.
In what follows, we denote ${K_{\mathrm{R}}\mu_{\mathrm{R}}}$ and ${K_{\mathrm{T}}\mu_{\mathrm{T}}}$ by $m_{\mathrm{C},\mathrm{R}}$
and $m_{\mathrm{C},\mathrm{T}}$ respectively.
%%%
From \cite{Naderializadeh2017}, we know that $\min\{m_{\mathrm{C},\mathrm{T}}+m_{\mathrm{C},\mathrm{R}}, K_{\mathrm{R}}  \}$ P-subpackets are delivered in each P-block, and hence
\begin{equation}
\nonumber
H_{\mathrm{C}}^{(\mathrm{p})}=\frac{K_{\mathrm{R}}F(1-\mu_{\mathrm{R}}) }{\min\{m_{\mathrm{C},\mathrm{T}}+m_{\mathrm{C},\mathrm{R}}, K_{\mathrm{R}}  \}}.
\end{equation}
On the other, we know from \cite{Maddah-Ali2014} that $\min\{1+m_{\mathrm{C},\mathrm{R}}, K_{\mathrm{R}}  \}$ N-subpackets are
delivered in each N-block. Therefore, we obtain
\begin{equation}
\nonumber
H_{\mathrm{C}}^{(\mathrm{n})}=\frac{K_{\mathrm{R}}F(1-\mu_{\mathrm{R}})}{\min\{1+m_{\mathrm{C},\mathrm{R}}, K_{\mathrm{R}}  \}}.
\end{equation}
It follows that the delivery time in time-slot is given by $ H_{\mathrm{C}} = \max \Big \{ \frac{q}{\alpha} H_{\mathrm{C}}^{(\mathrm{p})}, \frac{\bar{q}}{\bar{\alpha}} H_{\mathrm{C}}^{(\mathrm{n})} \Big \}$.
Next, we choose the splitting ratio $q$ as follow:
\begin{equation}
\nonumber
q=\frac{\alpha \cdot \min \{m_{\mathrm{C},\mathrm{T}}+m_{\mathrm{C},\mathrm{R}},K_{\mathrm{R} } \}}{  \alpha \cdot \min \{ m_{\mathrm{C},\mathrm{T}}+m_{\mathrm{C},\mathrm{R}},K_{\mathrm{R} } \}  + \bar{\alpha} \cdot \min \{1+m_{\mathrm{C},\mathrm{R}},K_{\mathrm{R}}\} }.
\end{equation}
It can be verified that the above splitting ratio satisfies
$\frac{q}{\alpha} H_{\mathrm{C}}^{(\mathrm{p})} = \frac{\bar{q}}{\bar{\alpha}} H_{\mathrm{C}}^{(\mathrm{n})}$.
This value of $q$ minimizes the duration of the communication which in turn maximizes the achievable DoF.
Note that $q$ increases with $\alpha$, due to the fact that a larger $\alpha$ implies that
the P-subchannel occupies a larger fraction of the bandwidth, hence carrying larger portions of each packet.
As one may anticipate, we obtain $q = 0$ and $q = 1$ at the two extremes $\alpha =0$
and $\alpha = 1$, respectively.
With such value of $q$ we obtain
%%%
\begin{equation}
\label{eq:centralized_delivery_time}
H_{\mathrm{C}} = \frac{K_{\mathrm{R}}F(1-\mu_{\mathrm{R}})}{ \alpha \cdot \min \{m_{\mathrm{C},\mathrm{T}}+m_{\mathrm{C},\mathrm{R}},K_{\mathrm{R}} \} + \bar{\alpha} \cdot  \min \{1+m_{\mathrm{C},\mathrm{R}},K_{\mathrm{R}}\}  }.
\end{equation}
From \eqref{eq:centralized_delivery_time} and the fact that a total of $K_{\mathrm{R}}F(1-\mu_{\mathrm{R}}) $ packets are delivered during the delivery phase, the result in \eqref{eq:achievable_DoF_theo_cen} directly follows.
This concludes the proof of achievability.
%%%%%%%%%%%%%%%%%%%%
\subsection{Converse of Theorem \ref{theorem_cen}}
To prove order optimality, we first derive an upper bound for the one-shot linear DoF.
\begin{lemma}
	\label{theorem_outerbound}
	For the cache-aided wireless network described in Section \ref{sec:problem_setting},
	the one-shot linear DoF of the network, defined in \eqref{eq:network_one_shot_linear_DoF}, is bounded above as
	\begin{equation}
	\label{eq:outerbound}
	\mathsf{DoF}_{\mathrm{L}}^*(\mu_{\mathrm{T}},\mu_{\mathrm{R}},\alpha) \leq \alpha \cdot \min \Big \{ \frac{K_{\mathrm{T}}\mu_{\mathrm{T}}+K_{\mathrm{R}}\mu_{\mathrm{R}}}{1-\mu_{\mathrm{R}}}, K_{\mathrm{R} } \Big \}  +\bar{\alpha} \cdot \min \Big \{\frac{1+K_{\mathrm{R}}\mu_{\mathrm{R}}}{1-\mu_{\mathrm{R}}},K_{\mathrm{R}}\Big \}.
	\end{equation}
\end{lemma}
The proof of Lemma \ref{theorem_outerbound} is relegated to Appendix \ref{sec:outerbound}.
%%%
It is easily seen that by denoting the right-hand side of \eqref{eq:outerbound} as $\mathsf{DoF}_{\mathrm{L}, \mathrm{ub}}(\mu_{\mathrm{T}},\mu_{\mathrm{R}},\alpha)$, we have
\begin{equation}
\label{eq:outerbound_lin}
\mathsf{DoF}_{\mathrm{L}, \mathrm{ub}} (\mu_{\mathrm{T}},\mu_{\mathrm{R}},\alpha) = \alpha \cdot  \mathsf{DoF}_{\mathrm{L}, \mathrm{ub}}(\mu_{\mathrm{T}},\mu_{\mathrm{R}},1) + \bar{\alpha} \cdot \mathsf{DoF}_{\mathrm{L}, \mathrm{ub}}(\mu_{\mathrm{T}},\mu_{\mathrm{R}},0).
\end{equation}
The expression in \eqref{eq:outerbound_lin} proofs useful when proving the order-optimality parts of Theorem \ref{theorem_cen}  and Theorem \ref{theorem_decen}.
We now proceed to prove the order-optimality part of Theorem \ref{theorem_cen}.

%%%
From \cite{Naderializadeh2017}, we know that for $\alpha=1$, we have
$\mathsf{DoF}_{\mathrm{L}, \mathrm{ub}}(\mu_{\mathrm{T}},\mu_{\mathrm{R}},1) / \mathsf{DoF}_{\mathrm{L}, \mathrm{C}}(\mu_{\mathrm{T}},\mu_{\mathrm{R}},1) \leq 2$.
%%%
We show that when $\alpha=0$, we also have $\mathsf{DoF}_{\mathrm{L}, \mathrm{ub}}(\mu_{\mathrm{T}},\mu_{\mathrm{R}},0) / \mathsf{DoF}_{\mathrm{L}, \mathrm{C}}(\mu_{\mathrm{T}},\mu_{\mathrm{R}},0) \leq 2$.
%%%
Consider the two cases:
\begin{enumerate}
	\item $\mu_{\mathrm{R}} \leq \frac{1}{2}$: In this case, from (\ref{eq:outerbound}) in Lemma \ref{theorem_outerbound}  we obtain
	\begin{equation}
\nonumber
	\begin{split}
	\mathsf{DoF}_{\mathrm{L}, \mathrm{ub}}(\mu_{\mathrm{T}},\mu_{\mathrm{R}}, 0) & =  \min \Big \{ \frac{1+K_{\mathrm{R}}\mu_{\mathrm{R}}}{1-\mu_{\mathrm{R}}}, K_{\mathrm{R} } \Big \} \\
	& \leq  \min \Big \{ \frac{1+K_{\mathrm{R}}\mu_{\mathrm{R}}}{1-1/2}, K_{\mathrm{R} } \Big \} \\
	& \leq 2 \cdot \mathsf{DoF}_{\mathrm{L}, \mathrm{C}}(\mu_{\mathrm{T}},\mu_{\mathrm{R}}, 0).
	\end{split}
	\end{equation}	
	\item $\mu_{\mathrm{R}} > \frac{1}{2}$:
	In this case, the achievability part implies that
	\begin{equation}
\nonumber
	\begin{split}
	\mathsf{DoF}_{\mathrm{L}, \mathrm{C}}(\mu_{\mathrm{T}},\mu_{\mathrm{R}}, 0) &= \min \{1+K_{\mathrm{R}}\mu_{\mathrm{R}},K_{\mathrm{R}}\} \\
	& > \min \{1 + K_{\mathrm{R}}/2 , K_{\mathrm{R}}\} > \frac{K_{\mathrm{R}}}{2}.
	\end{split}
	\end{equation}
	Since $\mathsf{DoF}_{\mathrm{L}, \mathrm{ub}}(\mu_{\mathrm{T}},\mu_{\mathrm{R}}, 0) \leq K_{\mathrm{R}}$, we obtain
$\mathsf{DoF}_{\mathrm{L}, \mathrm{ub}}(\mu_{\mathrm{T}},\mu_{\mathrm{R}}, 0)\leq 2 \cdot \mathsf{DoF}_{\mathrm{L}, \mathrm{C}}(\mu_{\mathrm{T}},\mu_{\mathrm{R}}, 0)$.	
\end{enumerate}
%%%%%
Now we extend the above to any $\alpha \in [0,1]$.
%%%
From the two above constant factor inequalities for $\alpha = 1$ and $\alpha = 0$, and
the decomposition of the lower bound and the upper bound in (\ref{eq:achievable_DoF_theo_cen_lin}) and  (\ref{eq:outerbound_lin}),
we obtain
\begin{equation}
\nonumber
\begin{split}
  \mathsf{DoF}_{\mathrm{L}, \mathrm{ub}}(\mu_{\mathrm{T}},\mu_{\mathrm{R}},\alpha) & =  \alpha \cdot  \mathsf{DoF}_{\mathrm{L}, \mathrm{ub}}(\mu_{\mathrm{T}},\mu_{\mathrm{R}},1) + \bar{\alpha} \cdot  \mathsf{DoF}_{\mathrm{L}, \mathrm{ub}}(\mu_{\mathrm{T}},\mu_{\mathrm{R}},0) \\
  &  \leq 2 \alpha \cdot  \mathsf{DoF}_{\mathrm{L}, \mathrm{C}}(\mu_{\mathrm{T}},\mu_{\mathrm{R}},1) + 2  \bar{\alpha} \cdot  \mathsf{DoF}_{\mathrm{L}, \mathrm{C}}(\mu_{\mathrm{T}},\mu_{\mathrm{R}},0) \\ & =  2 \cdot \mathsf{DoF}_{\mathrm{L}, \mathrm{C}}(\mu_{\mathrm{T}},\mu_{\mathrm{R}},\alpha).
\end{split}
\end{equation}
This completes the proof of Theorem \ref{theorem_cen}.
%%%%%%%%%%%%%%
\section{Decentralized Setting: Proof of Theorem \ref{theorem_decen}}
%%%%%%%%%%%%%%
\label{sec:decentralized_setting}
%%%%%%%%%%%%
In this section, we present a proof of Theorem \ref{theorem_decen} starting with the achievability and then the converse.
%%%%%%%%%%%%
\subsection{Achievability of Theorem \ref{theorem_decen}}
\subsubsection{Placement Phase}
As in the centralized setting, the placement phase does not depend on $\alpha$.
Each file $\mathcal{W}_n, n \in [N]$, is partitioned into  $\binom{K_{\mathrm{T}}}{K_{\mathrm{T}}\mu_{\mathrm{T}}}$ disjoint subfiles of equal size, denoted by
$\mathcal{W}_n=\{\mathcal{W}_{n,\mathcal{T}}\}_{\substack{\mathcal{T} \subseteq  [K_{\mathrm{T}}] : |\mathcal{T}|= {K_{\mathrm{T}}\mu_{\mathrm{T}}} }}$, where each subfile contains $\frac{F}{\binom{K_{\mathrm{T}}}{K_{\mathrm{T}}\mu_{\mathrm{T}}}}$ packets. Each transmitter $\text{Tx}_i$ then stores subfile given by $\mathcal{P}_i = \{ \mathcal{W}_{n,\mathcal{T}} : i \in \mathcal{T} \}$.
On the other end, placement at the receivers is done in a decentralized manner similar to \cite{Maddah-Ali2015a}.
%%%
In particular, each receiver $\text{Rx}_i$ stores $\mu_{\mathrm{R}}F$ packets from each
file, chosen uniformly at random.
Therefore, each packet of each file is stored in some subset of users $\tilde{\mathcal{R}} \subseteq [K_{\mathrm{R}}]$,
where $|\tilde{\mathcal{R}}| \in \{0,1,\dots,K_{\mathrm{R}}\}$.
For any $n \in [N]$, we use $\mathcal{W}_{n,\mathcal{T},\tilde{\mathcal{R}}}$ to denote the packets of file $\mathcal{W}_n$ which are stored by transmitters in $ \mathcal{T}$ and receivers in $\tilde{\mathcal{R}}$, where $\mathcal{W}_{n,\mathcal{T},\tilde{\mathcal{R}}}$ is referred to as a
mini-subfile henceforth.
It follows that $\mathcal{W}_{n}$ can be reconstructed from  $\big\{\mathcal{W}_{n,\mathcal{T},\tilde{\mathcal{R}}} : \mathcal{T} \subseteq [K_{\mathrm{T}}], |\mathcal{T}|= K_{\mathrm{T}}\mu_{\mathrm{T}}, \tilde{\mathcal{R}} \subseteq [K_{\mathrm{R}}] \big\}$.
\subsubsection{Delivery Phase}
%%
%%%%%%
Each receiver $\text{Rx}_j$ requests for a file $\mathcal{W}_{d_j}$, hence the transmitters have to deliver all mini-subfiles
$\mathcal{W}_{d_j, \mathcal{T},\tilde{\mathcal{R}}}$ with $j \notin \tilde{\mathcal{R}}$.
Each packet to be delivered is split as in the centralized case, and we use  $\mathcal{W}_{d_j, \mathcal{T}}^{(\mathrm{p})}$ (P-subfile) and $\mathcal{W}_{d_j, \mathcal{T}}^{(\mathrm{n})}$ (N-subfile) to denote the sets of P-subpackets and N-subpackets of $\mathcal{W}_{d_j, \mathcal{T}}$, respectively.
%%%%
Similarly, we use $\mathcal{W}_{d_j, \mathcal{T},\tilde{\mathcal{R}}}^{(\mathrm{p})}$ (P-mini-subfile) and $\mathcal{W}_{d_j, \mathcal{T},\tilde{\mathcal{R}}}^{(\mathrm{n})}$ (N-mini-subfile) to denote the sets of P-subpackets and N-subpackets of
$\mathcal{W}_{d_j, \mathcal{T},\tilde{\mathcal{R}}}$, respectively.
%%%%

The P-mini-subfiles are delivered over the P-subchannel, where the delivery takes place over $K_{\mathrm{R}}$ sub-phases indexed by $l \in \{0,1,\dots,K_{\mathrm{R}}-1\}$.
In the $l$-th sub-phase, the transmitters delivers all $\mathcal{W}_{d_j, \mathcal{T},\tilde{\mathcal{R}}}^{(\mathrm{p})}$ with $|\tilde{\mathcal{R}}|=l$.
Note that $l$ goes up to $K_{\mathrm{R}}-1$ since for $|\tilde{\mathcal{R}}|= K_{\mathrm{R}}$, the corresponding P-mini-subfiles are cached by all receivers.
For each sub-phase $l$, the delivery in the P-subchannel is reminiscent of the centralized P-subchannel delivery in Section \ref{sec:delivery_centralized}, with the difference that $m_{\mathrm{C},\mathrm{R}}$ in the centralized setting is replaced with $l$ here (i.e. smaller multicasting gain), as this sub-phase considers subfiles which are cached by exactly $l$ users.
It follows that $\min \{m_{\mathrm{C},\mathrm{T}}+l,K_{\mathrm{R}}\}$ P-subpackets are transmitted simultaneously.
%%%

%%%%
On the other hand, the N-mini-subfiles are delivered over the N-subchannel using the original decentralized coded-multicasting scheme in  \cite{Maddah-Ali2015a}, while using over the air superposition instead of XoR.
%%%
Each receiver then obtains all missing mini-subfiles and recovers the demanded file.
%%%
\subsubsection{Achievable One-Shot Linear DoF}
%%%%%%%%%%%%%%%%%%%%%%%%%%%%%%
We start be focusing on the delivery time over the P-subchannel.
%%%
Consider the $l$-th sub-phase and an arbitrary subset of users $\tilde{\mathcal{R}}$ with size $l$.
For each P-subfile $\mathcal{W}_{n,\mathcal{T}}^{(\mathrm{p})}$, $n \in [N]$, stored by some subset $\mathcal{T}$ of users),
the probability that any of its P-subpackets is stored by any of the users in $\tilde{\mathcal{R}}$ is given by $\mu_{\mathrm{R}}$,
as each such user caches $\mu_{\mathrm{R}}F$ random P-subpackets from each file.
%%%
Hence, the probability that a P-subpacket is stored by exactly the $l$ users of  $\tilde{\mathcal{R}}$ is given by $\mu_{\mathrm{R}}^{l}(1-\mu_{\mathrm{R}})^{K_{\mathrm{R}}-l}$.
%%%
It follows that the expected number of P-subpackets of $\mathcal{W}_{n,\mathcal{T}}^{(\mathrm{p})}$ stored by each user in
$\tilde{\mathcal{R}}$ is given by
%%%
$\frac{\mu_{\mathrm{R}}^{l}(1-\mu_{\mathrm{R}})^{K_{\mathrm{R}}-l}F}{\binom{K_{\mathrm{T}}}{K_{\mathrm{T}}\mu_{\mathrm{T}}}} + o(F)$ when $F \rightarrow \infty$. The term $o(F)$ is omitted henceforth.
%%%%
As there is a total of $\binom{K_{\mathrm{R}}}{l}$ subsets of $l$ users, there is a total of  $\frac{\binom{K_{\mathrm{R}}}{l}\mu_{\mathrm{R}}^{l}(1-\mu_{\mathrm{R}})^{K_{\mathrm{R}}-l}F}{\binom{K_{\mathrm{T}}}{K_{\mathrm{T}}\mu_{\mathrm{T}}}}$ P-subpackets of $\mathcal{W}_{n,\mathcal{T}}^{(\mathrm{p})}$ which are cached by exactly  $l$ users.
%%%%%%%%%%%%%%%%%%%
We now proceed to calculate  number of P-subpackets of $\mathcal{W}_{d_j}^{(\mathrm{p})}$ stored by exactly $l$ users and
have to be delivered to receiver $\text{Rx}_j$.
%%%%%%%%%
For each $\mathcal{T}$, receiver $\text{Rx}_j$ has all P-mini-subfiles $\mathcal{W}_{d_j,\mathcal{T},\tilde{\mathcal{R}}}^{(\mathrm{p})}$, with $|\tilde{\mathcal{R}}|=l$ and $j \in \tilde{\mathcal{R}}$, cached in its memory.
%%%
Hence, $\text{Rx}_j$ already has $\frac{\binom{K_{\mathrm{R}}-1}{l-1}\mu_{\mathrm{R}}^{l}(1-\mu_{\mathrm{R}})^{K_{\mathrm{R}}-l}F}{\binom{K_{\mathrm{T}}}
{K_{\mathrm{T}}\mu_{\mathrm{T}}}}$
%%%%
P-subpackets of $\mathcal{W}_{d_j,\mathcal{T}}^{(\mathrm{p})}$ which are cached by exactly $l$ users.
%%%
It follows that the number of P-subpackets of $\mathcal{W}_{d_j,\mathcal{T}}^{(\mathrm{p})}$ unavailable at $\text{Rx}_j$,
given by all P-mini-subfiles $\mathcal{W}_{d_j,\mathcal{T},\tilde{\mathcal{R}}}^{(\mathrm{p})}$ with $|\tilde{\mathcal{R}}|=l$ and $j \notin \tilde{\mathcal{R}}$, is equal to $\frac{\binom{K_{\mathrm{R}}-1}{l}\mu_{\mathrm{R}}^{l}(1-\mu_{\mathrm{R}})^{K_{\mathrm{R}}-l}F}{\binom{K_{\mathrm{T}}}{K_{\mathrm{T}}\mu_{\mathrm{T}}}}$.
%%%%%%%%%
Considering all possible P-subfiles $\mathcal{W}_{d_j,\mathcal{T}}^{(\mathrm{p})}$ for all $\mathcal{T}$, and as there are $K_{\mathrm{R}}$ receivers in total, the total number of P-subpackets
which are stored by exactly $l$ users and have to be delivered to all receivers in the $l$-th delivery sub-phase is given by
\begin{equation}
\nonumber
K_{\mathrm{R}} \binom{K_{\mathrm{R}}-1}{l}\mu_{\mathrm{R}}^{l}(1-\mu_{\mathrm{R}})^{K_{\mathrm{R}}-l}F.
\end{equation}
%%%
We recall that in the $l$-th delivery sub-phase, a total of $\min \{m_{\mathrm{C},\mathrm{T}}+l,K_{\mathrm{R}}\}$ P-subpackets are delivered simultaneously over the P-subchannel. By summing over all $K_{\mathrm{R}}$ sub-phases, we obtain
\begin{equation}
\nonumber
H^{(\mathrm{p})}_{\mathrm{D}} = K_{\mathrm{R}} \sum_{l=0}^{K_{\mathrm{R}}-1}{\frac{ \binom{K_{\mathrm{R}}-1}{l}\mu_{\mathrm{R}}^{l}(1-\mu_{\mathrm{R}})^{K_{\mathrm{R}}-l}F}{\min \{m_{\mathrm{C},\mathrm{T}}+l,K_{\mathrm{R}}\}}}.
\end{equation}
%%%%%%%%%%%

%%%
Moving on to the N-subchannel, as the delivery of the N-mini-subfiles follows the coded-multicasting scheme of \cite{Maddah-Ali2015a}, it follows that
%%%%
\begin{equation}
\nonumber
H_{\mathrm{D}}^{(\mathrm{n})} = K_{\mathrm{R}} \sum_{l=0}^{K_{\mathrm{R}}-1}{\frac{ \binom{K_{\mathrm{R}}-1}{l}\mu_{\mathrm{R}}^{l}(1-\mu_{\mathrm{R}})^{K_{\mathrm{R}}-l}F}{1+l}} = \frac{1- \mu_{\mathrm{R}}}{\mu_{\mathrm{R}}} \left( 1 - (1-\mu_{\mathrm{R}})^{K_{\mathrm{R}}} \right)F.
\end{equation}
From the above, it follows that the delivery time is given by $ H_{\mathrm{D}} = \max \Big \{ \frac{q}{\alpha} H^{(\mathrm{p})}_{\mathrm{D}}, \frac{\bar{q}}{\bar{\alpha}} H^{(\mathrm{n})}_{\mathrm{D}} \Big \}$ time-slots.
%%%%
As for the centralized case, we choose $q$ such that
$\frac{q}{\alpha} H^{(\mathrm{p})}_{\mathrm{D}} = \frac{\bar{q}}{\bar{\alpha}} H^{(\mathrm{n})}_{\mathrm{D}}$, which in turn minimizes the duration of the communication and hence maximizes the achievable DoF.
Hence, we choose
\begin{equation}
\nonumber
q= \frac{\alpha \cdot \frac{1}{\sum_{l=0}^{K_{\mathrm{R}}-1}
		\frac{\binom{K_{\mathrm{R}}-1}{l} \mu_{\mathrm{R}}^l (1-\mu_{\mathrm{R}})^{K_{\mathrm{R}}-l-1}}{\min \{ K_{\mathrm{R}} ,K_{\mathrm{T}}{\mu}_{\mathrm{T}}+l \} }}}{
	\alpha \cdot \frac{1}{\sum_{l=0}^{K_{\mathrm{R}}-1} \frac{\binom{K_{\mathrm{R}}-1}{l} \mu_{\mathrm{R}}^l (1-\mu_{\mathrm{R}})^{K_{\mathrm{R}}-l-1}}{\min \{ K_{\mathrm{R}} ,K_{\mathrm{T}}{\mu}_{\mathrm{T}}+l \}}} +	\bar{\alpha} \cdot \frac{K_{\mathrm{R}}\mu_{\mathrm{R}}}{1-(1-\mu_{\mathrm{R}})^{K_{\mathrm{R}}}}
}.
\end{equation}
From the above choice of $q$ and the values of $H^{(\mathrm{p})}_{\mathrm{D}}$ and $H^{(\mathrm{p})}_{\mathrm{C}}$,
it follows that
\begin{equation}
\label{eq:decentralized_delivery_time}
H_{\mathrm{D}} = \frac{K_{\mathrm{R}}F(1-\mu_{\mathrm{R}})}{ \alpha \cdot \frac{1}{\sum_{l=0}^{K_{\mathrm{R}}-1} \frac{\binom{K_{\mathrm{R}}-1}{l} \mu_{\mathrm{R}}^l (1-\mu_{\mathrm{R}})^{K_{\mathrm{R}}-l-1}}{\min \{m_{\mathrm{C},\mathrm{T}}+l,K_{\mathrm{R}}\}}} + \bar{\alpha} \frac{K_{\mathrm{R}}\mu_{\mathrm{R}}}{   \ 1 - (1-\mu_{\mathrm{R}})^{K_{\mathrm{R}}} }.
	}
\end{equation}
As a total of $K_{\mathrm{R}}F(1-\mu_{\mathrm{R}}) $ packets are delivered during the delivery phase, the result in \eqref{eq:achievable_DoF_theo_decen} directly follows from \eqref{eq:decentralized_delivery_time}, which concludes the proof of achievability.
\subsection{Converse of Theorem \ref{theorem_decen}}
%%%%%%%%%%%%%
In this part, we prove \eqref{eq:converse_theo_decen} through the following steps:
%%%
\begin{itemize}
\item The first step of the proof is to show that when $\alpha = 0$, we have the constant factor
\begin{equation} \label{eq:ineq_lemma_con_dece_1}
\frac{\mathsf{DoF}_{\mathrm{L}, \mathrm{ub}}(\mu_{\mathrm{T}},\mu_{\mathrm{R}},0)} {\mathsf{DoF}_{\mathrm{L}, \mathrm{D}}(\mu_{\mathrm{T}},\mu_{\mathrm{R}},0)} \leq 3.
\end{equation}
%%%
\item The following step is to show that the one-shot linear DoF ratio in \eqref{eq:ineq_lemma_con_dece_1}, with $\alpha = 0$, is an upper bound for the ratio  with $\alpha = 1$, i.e.
\begin{equation} \label{eq:ineq_lemma_con_dece_2}
\frac{\mathsf{DoF}_{\mathrm{L}, \mathrm{ub}}(\mu_{\mathrm{T}},\mu_{\mathrm{R}},1)} {\mathsf{DoF}_{\mathrm{L}, \mathrm{D}}(\mu_{\mathrm{T}},\mu_{\mathrm{R}},1)} \leq
\frac{\mathsf{DoF}_{\mathrm{L}, \mathrm{ub}}(\mu_{\mathrm{T}},\mu_{\mathrm{R}},0)} {\mathsf{DoF}_{\mathrm{L}, \mathrm{D}}(\mu_{\mathrm{T}},\mu_{\mathrm{R}},0)}.
\end{equation}
\item Equipped with \eqref{eq:ineq_lemma_con_dece_1} and \eqref{eq:ineq_lemma_con_dece_2}, we proceed ad follows:
\begin{equation}
\nonumber
\begin{split}
\mathsf{DoF}_{\mathrm{L}, \mathrm{ub}}(\mu_{\mathrm{T}},\mu_{\mathrm{R}},\alpha) & =  \alpha \cdot  \mathsf{DoF}_{\mathrm{L}, \mathrm{ub}}(\mu_{\mathrm{T}},\mu_{\mathrm{R}},1) + \bar{\alpha} \cdot \mathsf{DoF}_{\mathrm{L}, \mathrm{ub}}(\mu_{\mathrm{T}},\mu_{\mathrm{R}},0) \\
&  \leq 3  \alpha \cdot  \mathsf{DoF}_{\mathrm{L}, \mathrm{D}}(\mu_{\mathrm{T}},\mu_{\mathrm{R}},1) + 3  \bar{\alpha} \cdot  \mathsf{DoF}_{\mathrm{L}, \mathrm{D}}(\mu_{\mathrm{T}},\mu_{\mathrm{R}},0) \\ & =  3 \cdot \mathsf{DoF}_{\mathrm{L}, \mathrm{D}}(\mu_{\mathrm{T}},\mu_{\mathrm{R}},\alpha).
\end{split}
\end{equation}
%%%%%%%%%%%%%
\end{itemize}
It can be seen that the last of the three above steps concludes the proof of Theorem \ref{theorem_decen}.
%%%%
Therefore, the remainder of this part is dedicated to proving the inequalities in \eqref{eq:ineq_lemma_con_dece_1} and \eqref{eq:ineq_lemma_con_dece_2}.
%%%%
\subsubsection{Proof of \eqref{eq:ineq_lemma_con_dece_1}}
%%%%
%%%
First, we recall that $\mathsf{DoF}_{\mathrm{L}, \mathrm{D}}(\mu_{\mathrm{T}},\mu_{\mathrm{R}},0) = \frac{K_{\mathrm{R}}\mu_{\mathrm{R}}}{1-(1-\mu_{\mathrm{R}})^{K_{\mathrm{R}}}}$.
Combining this with $(1-\mu_{\mathrm{R}})^{K_{\mathrm{R}}} \geq 0$ and the Bernoulli inequality
$(1-\mu_{\mathrm{R}})^{K_{\mathrm{R}}} \geq 1- K_{\mathrm{R}} \mu_{\mathrm{R}} $, we obtain
\begin{equation}
\label{eq:DoF_LD_LB_alph_0}
\mathsf{DoF}_{\mathrm{L}, \mathrm{D}}(\mu_{\mathrm{T}},\mu_{\mathrm{R}},0) \geq \max
\Big \{ K_{\mathrm{R}}\mu_{\mathrm{R}}, 1 \Big \}.
\end{equation}	
%%%	
For the trivial case of $K_{\mathrm{R}}=1$, it is easy to see that
$\mathsf{DoF}_{\mathrm{L}, \mathrm{D}}(\mu_{\mathrm{T}},\mu_{\mathrm{R}},0) = \mathsf{DoF}_{\mathrm{L}, \mathrm{ub}}(\mu_{\mathrm{T}},\mu_{\mathrm{R}},0)=1$.
%%%%
For the case of $K_{\mathrm{R}}=2$, we have $\mathsf{DoF}_{\mathrm{L}, \mathrm{D}}(\mu_{\mathrm{T}},\mu_{\mathrm{R}},0) \geq 1$ from \eqref{eq:DoF_LD_LB_alph_0} and $\mathsf{DoF}_{\mathrm{L}, \mathrm{ub}}(\mu_{\mathrm{T}},\mu_{\mathrm{R}},0) \leq 2$ from \eqref{eq:outerbound}
in Lemma \ref{theorem_outerbound}.
%%%%
Hence for this case, \eqref{eq:ineq_lemma_con_dece_1} holds.
%%%%
Similarly, for the case $K_{\mathrm{R}}=3$, we have
$\mathsf{DoF}_{\mathrm{L}, \mathrm{D}}(\mu_{\mathrm{T}},\mu_{\mathrm{R}},0) \geq 1$ and
$\mathsf{DoF}_{\mathrm{L}, \mathrm{ub}}(\mu_{\mathrm{T}},\mu_{\mathrm{R}},0) \leq 3$
from which \eqref{eq:ineq_lemma_con_dece_1} also holds.
%%%%
Therefore, without loss of generality, we assume that $K_{\mathrm{R}} \geq 4$ henceforth.
%%%%
We proceed by considering the following cases:
%%%%
\begin{enumerate}
\item %%
$ \mu_{\mathrm{R}} \leq  {1}/{K_{\mathrm{R}}}$:
For this case we have
\begin{equation}
\nonumber
\begin{split}
\mathsf{DoF}_{\mathrm{L}, \mathrm{ub}}(\mu_{\mathrm{T}},\mu_{\mathrm{R}},0) & = \min \Big \{  \frac{K_{\mathrm{R}}\mu_{\mathrm{R}}+1}{1-\mu_{\mathrm{R}}}, K_{\mathrm{R}}  \Big \} \\
& \leq \min \Big \{  \frac{1+1}{1-1/K_{\mathrm{R}}}, K_{\mathrm{R}}  \Big \} \\
& \leq \min \Big \{  \frac{8}{3}, K_{\mathrm{R}}  \Big \} \leq 3.
\end{split}
\end{equation}
%%%%%
Combining the above with $\mathsf{DoF}_{\mathrm{L}, \mathrm{D}}(\mu_{\mathrm{T}},\mu_{\mathrm{R}},0) \geq 1$,
we conclude that \eqref{eq:ineq_lemma_con_dece_1} holds.
%%%
\item  $\mu_{\mathrm{R}} \in (1/K_{\mathrm{R}},2/K_{\mathrm{R}}]$:
For this case, we start by defining the function
	\begin{equation}
    \nonumber
	f(\mu_{\mathrm{R}})=3 \mu_{\mathrm{R}} + \frac{1}{K_{\mathrm{R}}\mu_{\mathrm{R}}}.
	\end{equation}
	The function $f(\mu_{\mathrm{R}})$ is convex in $[0,\infty)$, and hence
    $f(\mu_{\mathrm{R}}) \leq \max \big( f(\frac{1}{K_{\mathrm{R}}}), f (\frac{2}{K_{\mathrm{R}}} ) \big)$ over the interval of interest  $\mu_{\mathrm{R}} \in (1/K_{\mathrm{R}},2/K_{\mathrm{R}}]$.
    Moreover, it is easy to verify that $f(\frac{1}{K_{\mathrm{R}}}) =  \frac{3}{K_{\mathrm{R}}} + 1 \leq \frac{7}{4}$
	and $f(\frac{2}{K_{\mathrm{R}}}) =  \frac{6}{K_{\mathrm{R}}} + \frac{1}{2} \leq 2$.
	Therefore, $f(\mu_{\mathrm{R}}) = 3 \mu_{\mathrm{R}} + \frac{1}{K_{\mathrm{R}}\mu_{\mathrm{R}}} \leq 2 $ for all $K_{\mathrm{R}}$ and $\mu_{\mathrm{R}}$ of interest.
Combining this with \eqref{eq:DoF_LD_LB_alph_0} and \eqref{eq:outerbound}, we obtain
	\begin{equation}
\nonumber
	\begin{split}
	\frac{\mathsf{DoF}_{\mathrm{L}, \mathrm{ub}}(\mu_{\mathrm{T}},\mu_{\mathrm{R}},0)}{\mathsf{DoF}_{\mathrm{L}, \mathrm{D}}(\mu_{\mathrm{T}},\mu_{\mathrm{R}},0)} & \leq
\min \Big \{  \frac{K_{\mathrm{R}}\mu_{\mathrm{R}}+1}{1-\mu_{\mathrm{R}}}, K_{\mathrm{R}}  \Big \}
\cdot \frac{1}{\max
 \{ K_{\mathrm{R}}\mu_{\mathrm{R}}, 1 \}} \\
&  \leq \Big(1 + \frac{1}{K_{\mathrm{R}}\mu_{\mathrm{R}}} \Big )\cdot \frac{1}{1-\mu_{\mathrm{R}}} \leq 3
	\end{split}
	\end{equation}
%%%
%%
where the last inequality is equivalent to $3 \mu_{\mathrm{R}} + \frac{1}{K_{\mathrm{R}}\mu_{\mathrm{R}}} \leq 2$.
%%%
Therefore, \eqref{eq:ineq_lemma_con_dece_1} holds in this case.
%%%%
\item $\mu_{\mathrm{R}} \in (2/K_{\mathrm{R}},1/2]$: For this case we have
	\begin{equation}
\nonumber
	\begin{split}
	\mathsf{DoF}_{\mathrm{L}, \mathrm{ub}}(\mu_{\mathrm{T}},\mu_{\mathrm{R}},0) & = \min \Big \{  \frac{K_{\mathrm{R}}\mu_{\mathrm{R}}+1}{1-\mu_{\mathrm{R}}}, K_{\mathrm{R}}  \Big \} \\
	& \leq \min \Big \{  \frac{K_{\mathrm{R}}\mu_{\mathrm{R}}+1}{1-1/2}, K_{\mathrm{R}} \Big  \} \\
	& = \min \Big \{ 2  K_{\mathrm{R}}\mu_{\mathrm{R}}+ 2, K_{\mathrm{R}} \Big  \} \\
	& \leq \min \Big \{ 3  K_{\mathrm{R}}\mu_{\mathrm{R}}, K_{\mathrm{R}} \Big  \}.
	\end{split}
	\end{equation}
	%%%
Combining the above with $\mathsf{DoF}_{\mathrm{L}, \mathrm{D}}(\mu_{\mathrm{T}},\mu_{\mathrm{R}},0) \geq  K_{\mathrm{R}}\mu_{\mathrm{R}}$,
it follows that \eqref{eq:ineq_lemma_con_dece_1} holds.
\item $\mu_{\mathrm{R}} > 1/2$: For this last case we have
$\mathsf{DoF}_{\mathrm{L}, \mathrm{D}}(\mu_{\mathrm{T}},\mu_{\mathrm{R}},0) \geq \max \{ K_{\mathrm{R}}\mu_{\mathrm{R}}, 1 \} >  K_{\mathrm{R}}/2$. Combining this with $\mathsf{DoF}_{\mathrm{L}, \mathrm{ub}}(\mu_{\mathrm{T}},\mu_{\mathrm{R}},0) \leq K_{\mathrm{R}} $, it follows that \eqref{eq:ineq_lemma_con_dece_1} holds, hence concluding the proof.
\end{enumerate}
%%%%
\subsubsection{Proof of \eqref{eq:ineq_lemma_con_dece_2}}
%%%%
From \eqref{eq:achievable_DoF_theo_decen} and \eqref{eq:outerbound}, the inequality in \eqref{eq:ineq_lemma_con_dece_2} can be expressed as
	%%%
	\begin{equation} \label{eq:ineq_to_prove}
	\frac{ \min \Big \{   \frac{1+ K_{\mathrm{R}} \mu_{\mathrm{R}} }{1-\mu_{\mathrm{R}}} , K_{\mathrm{R}}      \Big \} }{ \left(
	\sum_{m=0}^{K_{\mathrm{R}}-1}	\frac{ \binom{K_{\mathrm{R}}-1}{m}
	\mu_{\mathrm{R}}^m (1- \mu_{\mathrm{R}})^{K_{\mathrm{R}}-1-m}}{1+m}
		\right)^{-1}  } \geq
	\frac{ \min \Big \{   \frac{K_{\mathrm{T}} \mu_{\mathrm{T}}+ K_{\mathrm{R}} \mu_{\mathrm{R}} }{1-\mu_{\mathrm{R}}} , K_{\mathrm{R}}      \Big \} }{ \left(
		\sum_{m=0}^{K_{\mathrm{R}}-1}	\frac{ \binom{K_{\mathrm{R}}-1}{m}
			\mu_{\mathrm{R}}^m (1- \mu_{\mathrm{R}})^{K_{\mathrm{R}}-1-m}}{
			\min \{ K_{\mathrm{T}} \mu_{\mathrm{T}}+m ,  K_{\mathrm{R}} \} } \right)^{-1} }.
	\end{equation}
	%%%%%%%%%%%%%
	Defining the function $J(r)$ as
	%%%%%%%%%
	\begin{equation} \label{eq:function_J}
	J(r)= \frac{ \min \Big \{   \frac{r+ K_{\mathrm{R}} \mu_{\mathrm{R}} }{1-\mu_{\mathrm{R}}} , K_{\mathrm{R}}      \Big \} }{  \left(
		\sum_{m=0}^{K_{\mathrm{R}}-1}	\frac{ \binom{K_{\mathrm{R}}-1}{m}
			\mu_{\mathrm{R}}^m (1- \mu_{\mathrm{R}})^{K_{\mathrm{R}}-1-m}}{
			\min \{ r+m, K_{\mathrm{R}} \} } \right)^{-1} }
	\end{equation}
	%%%%%%%%%%%%
    it can be seen that \eqref{eq:ineq_to_prove} is equivalent to $J(1) \geq J(K_{\mathrm{T}} \mu_{\mathrm{T}}) $.
    %%%%
    In the following, we show that that $J(1) \geq J(r)$ for all $r \geq 1$.
	%%%%%
	As a consequence, $J(1) \geq J(r)$ will also hold for integer values of $r$,
    hence for any $K_{\mathrm{T}}\mu_{\mathrm{T}}$ which is assumed to be integer for the decentralized setting and hence in Theorem \ref{theorem_decen} and in \eqref{eq:ineq_to_prove}.
	%%%%%
	%We consider here the case of $r$ integer
	%as $ K_{\mathrm{T}} \mu_{\mathrm{T}}$ is assumed integer, however the %same steps can be used to show the result for any $r \geq 1$.
	%%%
	
    %%%
	It is readily seen that for $r \geq K_{\mathrm{R}}(1-2\mu_{\mathrm{R}})$,
    the numerator in \eqref{eq:function_J} becomes $K_{\mathrm{R}}$, and the function $J(r)$ decrease with $r$.
    Therefore, without loss of generality, we only consider the interval $r \in [1,K_{\mathrm{R}}(1-2\mu_{\mathrm{R}})]$ in what follows.
    Equivalently, for any $K_{\mathrm{R}}$  and $r$, we consider values of $\mu_{\mathrm{R}}$ that satisfy $\mu_{\mathrm{R}} \leq \frac{1}{2} \left(1 - \frac{r}{K_{\mathrm{R}}}\right)$.

Next, the inequality in \eqref{eq:ineq_to_prove} is equivalently rewritten as
	%%%%%
	%\vspace{-2pt}
	\begin{equation}
\nonumber
	\frac{1+K_{\mathrm{R}} \mu_{\mathrm{R}}}{1- \mu_{\mathrm{R}}}   \sum_{m=0}^{K_{\mathrm{R}}-1} {\binom{K_{\mathrm{R}}-1}{m}} \frac{\mu_{\mathrm{R}}^m (1- \mu_{\mathrm{R}})^{K_{\mathrm{R}}-1-m}}{1+m}
	\geq
	\frac{r+K_{\mathrm{R}} \mu_{\mathrm{R}}}{1- \mu_{\mathrm{R}}}    \sum_{m=0}^{K_{\mathrm{R}}-1} {\binom{K_{\mathrm{R}}-1}{m}} \frac{\mu_{\mathrm{R}}^m (1- \mu_{\mathrm{R}})^{K_{\mathrm{R}}-1-m}}{\min \{ r+m ,  K_{\mathrm{R}} \}}.
	%\vspace{-2pt}
	\end{equation}
	%%%%
After rearranging the terms and removing redundant factors, the above is expressed as
	%%
	%\vspace{-2pt}
	\begin{equation}
\nonumber
	 \sum_{m=0}^{K_{\mathrm{R}}-1}  \frac{{1+K_{\mathrm{R}} \mu_{\mathrm{R}}}}{1+m} {\binom{K_{\mathrm{R}}-1}{m}} \left(\frac{\mu_{\mathrm{R}}}{1- \mu_{\mathrm{R}}}\right)^m
     \geq
     \sum_{m=0}^{K_{\mathrm{R}}-1}  \frac{{r+K_{\mathrm{R}} \mu_{\mathrm{R}}}}{\min \{ r+m ,  K_{\mathrm{R}} \}} {\binom{K_{\mathrm{R}}-1}{m}} \left(\frac{\mu_{\mathrm{R}}}{1- \mu_{\mathrm{R}}}\right)^m,
     %\vspace{-2pt}
	\end{equation}
	%%%%%%
	which is further rewritten as
	%%%%%%%%%%
	%\vspace{-2pt}
	\begin{equation} \label{eq:ineq_zeta_poly}
     \sum_{m=0}^{K_{\mathrm{R}}-1} \frac{\zeta(K_{\mathrm{R}}+1)+1}{1+m} {\binom{K_{\mathrm{R}}-1}{m}} \zeta^m
	  \geq
	  \sum_{m=0}^{K_{\mathrm{R}}-1} \frac{\zeta(K_{\mathrm{R}}+r)+r}{\min \{ r+m ,  K_{\mathrm{R}} \}} {\binom{K_{\mathrm{R}}-1}{m}} \zeta^m,
	  %\vspace{-2pt}
	\end{equation}
	%%%
where $\zeta=\frac{\mu_{\mathrm{R}}}{1-\mu_{\mathrm{R}}}$, which is constrained as $\zeta \in \left[0,\frac{K_{\mathrm{R}}-r}{K_{\mathrm{R}}+r}\right]$ for given $K_{\mathrm{R}}$ and $r$.
	%%%%%%%%%%%%%%%%%
After further rearrangement of terms, the inequality in (\ref{eq:ineq_zeta_poly}) is rewritten as
%%%%
%\vspace{-5pt}
\begin{equation} \label{eq:ineq_polynomial}
p(\zeta) \triangleq \sum_{m=0}^{K_{\mathrm{R}}} {c_m \cdot \zeta^m} \geq 0,
%\vspace{-3pt}
\end{equation}
%%%
where $p(\zeta)$ is a polynomial in the variable $\zeta$ with coefficients given by
\begin{equation}
\nonumber
 c_m=
 \begin{cases}
  0, \quad m=0\\
 \frac{1-r}{K_{\mathrm{R}}}, \quad m=K_{\mathrm{R}} \\
 \binom{K_{\mathrm{R}}-1}{m-1} \cdot \left( \frac{K_{\mathrm{R}}+1}{m} - \frac{K_{\mathrm{R}}+r}{\min \{ r+m -1 ,  K_{\mathrm{R}} \}} \right)  + \binom{K_{\mathrm{R}}-1}{m} \cdot \left( \frac{1}{m+1} - \frac{r}{\min \{ r+m ,  K_{\mathrm{R}} \}} \right) , \quad m \in [1,  K_{\mathrm{R}}-1 ]_{\mathbb{Z}}.\\
 \end{cases}
\end{equation}
    %%%
Note that in the above, we use $[a,b]_{\mathbb{Z}}$ to denote the set of all integers that are in the interval $[a,b]$, i.e. $[a,b]_{\mathbb{Z}} \triangleq [a,b] \cap \mathbb{Z}$.
   %%%
At this point, it is clear that the problem reduces to showing that  $p(\zeta) \geq 0$ for $\zeta \in \left[0,\frac{K_{\mathrm{R}}-r}{K_{\mathrm{R}}+r}\right]$.
To this end, we derive the following property of  $p(\zeta)$.
\begin{lemma}
\label{lemma:poly_con_decen}
The polynomial $p(\zeta)$ is quasiconcave and hence satisfies the following inequality:
\begin{equation}
 \label{eq:poly_con_decen}
    	  p(\zeta) \geq \min \left(  p(0) , p \left( \frac{K_{\mathrm{R}}-r}{K_{\mathrm{R}}+r} \right) \right),
  \ \forall  \zeta \in \left[0, \frac{K_{\mathrm{R}}-r}{K_{\mathrm{R}}+r} \right].
 \end{equation}
\end{lemma}
%%%
The proof of \eqref{eq:poly_con_decen} is rather involved and hence is deferred to  Appendix \ref{sec:proof_poly_property}.
%%%
From Lemma \ref{lemma:poly_con_decen},
it follows that to prove that the inequality in \eqref{eq:ineq_polynomial} holds,
it is sufficient to show that $p(0) \geq 0$ and $p \left( \frac{K_{\mathrm{R}}-r}{K_{\mathrm{R}}+r} \right) \geq 0$.
Note that the case with $\zeta=0$ is trivial as $p(0)=0$.
Hence, it remains to show that $p \left( \frac{K_{\mathrm{R}}-r}{K_{\mathrm{R}}+r} \right) \geq 0$ holds true.
%%%
For this, we require the following inequality.
%%%
\begin{lemma}
\label{lemma:inequality}
\cite{Pinelis2018}.
For any positive integer $K \in \mathbb{Z}_{+}$ and real number $r \in [1,K]$, we have
\begin{equation} \label{eq:ineq_lemma}
\sum_{m=1}^{K}  \frac{m}{\min \{ r+m-1 ,  K \}} {\binom{K}{m}} \left(\frac{K-r}{K+r} \right)^{m} \leq
\frac{K-r+2}{K+r}  \left[ \left(\frac{2K}{K+r} \right)^{K} -1 \right].
\end{equation}
\end{lemma}
%%%
The final step of the proof is to show that the inequality $p \left( \frac{K_{\mathrm{R}}-r}{K_{\mathrm{R}}+r} \right) \geq 0$
is an instance of Lemma \ref{lemma:inequality}, and hence holds true.
%%%
Equivalently, we consider \eqref{eq:ineq_zeta_poly}. By plugging $\zeta=\frac{K_{\mathrm{R}}-r}{K_{\mathrm{R}}+r}$ into \eqref{eq:ineq_zeta_poly}  and multiplying both sides by $\frac{K_{\mathrm{R}} + r}{K_{\mathrm{R}}}$, the inequality $p \left( \frac{K_{\mathrm{R}}-r}{K_{\mathrm{R}}+r} \right) \geq 0$ is equivalently expressed as
%%%%
%\vspace{-2pt}
\begin{equation}
\nonumber
	\sum_{m=0}^{K_{\mathrm{R}}-1}  \frac{K_{\mathrm{R}}-r+2}{1+m}  {\binom{K_{\mathrm{R}}-1}{m}} \left(\frac{K_{\mathrm{R}}-r}{K_{\mathrm{R}}+r} \right)^m
	\geq
	\sum_{m=0}^{K_{\mathrm{R}}-1}  \frac{K_{\mathrm{R}}+r}{\min \{ r+m ,  K_{\mathrm{R}} \}} {\binom{K_{\mathrm{R}}-1}{m}} \left(\frac{K_{\mathrm{R}}-r}{K_{\mathrm{R}}+r} \right)^m.
	%\vspace{-2pt}
	\end{equation}
By rearranging the above inequality and using the fact that  $\binom{K_{\mathrm{R}}}{m+1}=\binom{K_{\mathrm{R}}-1}{m} \frac{K_{\mathrm{R}}}{m+1}$,
we obtain
	%%%
	%\vspace{-2pt}
\begin{equation}
\label{eq:before_final_ineq}
	\frac{K_{\mathrm{R}}-r+2}{K_{\mathrm{R}}+r}   \sum_{m=1}^{K_{\mathrm{R}}}     {\binom{K_{\mathrm{R}}}{m}} \left(\frac{K_{\mathrm{R}}-r}{K_{\mathrm{R}}+r} \right)^{m}
	\geq
	\sum_{m=0}^{K_{\mathrm{R}}-1}  \frac{K_{\mathrm{R}}}{\min \{ r+m ,  K_{\mathrm{R}} \}} {\binom{K_{\mathrm{R}}-1}{m}} \left(\frac{K_{\mathrm{R}}-r}{K_{\mathrm{R}}+r} \right)^{m+1}.
	%\vspace{-2pt}
	\end{equation}
	%%%%
	By employing $\binom{K_{\mathrm{R}}}{m+1}=\binom{K_{\mathrm{R}}-1}{m} \frac{K_{\mathrm{R}}}{m+1}$ one more time, we finally arrive at
	%%%%
	%\vspace{-2pt}
	\begin{equation} \label{eq:final_ineq}
	\frac{K_{\mathrm{R}}-r+2}{K_{\mathrm{R}}+r}  \left[ \left(\frac{2K_{\mathrm{R}}}{K_{\mathrm{R}}+r} \right)^{K_{\mathrm{R}}} -1 \right]
	\geq
	\sum_{m=1}^{K_{\mathrm{R}}}  \frac{m}{\min \{ r+m-1 ,  K_{\mathrm{R}} \}} {\binom{K_{\mathrm{R}}}{m}} \left(\frac{K_{\mathrm{R}}-r}{K_{\mathrm{R}}+r} \right)^{m}.
	%\vspace{-2pt}
	\end{equation}
	%%%%%%%
where in going from \eqref{eq:before_final_ineq} to \eqref{eq:final_ineq}, we used the binomial identity to obtain $\sum_{m=1}^{K_{\mathrm{R}}}     {\binom{K_{\mathrm{R}}}{m}} \left(\frac{K_{\mathrm{R}}-r}{K_{\mathrm{R}}+r} \right)^{m} =  \left(\frac{2K_{\mathrm{R}}}{K_{\mathrm{R}}+r} \right)^{K_{\mathrm{R}}} -1$.
At this point, it is evident that the inequality in \eqref{eq:final_ineq} holds true due to \eqref{eq:ineq_lemma} in Lemma \ref{lemma:inequality}.
%%%
Therefore, \eqref{eq:ineq_polynomial} holds and the proof of \eqref{eq:ineq_lemma_con_dece_2} is complete.
%%%%%%%%%%%%%%
\section{Conclusions} \label{sec:conclusion}
%%%%%%%%%%%%%%
In the paper, we considered the problem of cache-aided interference management in
a wireless network where each node is equipped with a cache memory and transmission occurs over two parallel channels,
one for which perfect CSIT is available and another for which no CSIT is available.
%%%
Focusing on strategies with uncoded placement and separable one-shot linear delivery schemes,
we characterized the optimum one-shot linear DoF to within a multiplicative factor of $2$.
%%%
We further considered a decentralized setting in which content caching at the receivers is randomized.
%%%
For this decentralized setting, we characterized the
optimum one-shot linear DoF to within a multiplicative factor of $3$.
%%%
Our results generalize and expand upon previous one-shot linear DoF results in literature, namely
\cite{Shariatpanahi2016} and \cite{Naderializadeh2017}, by including the parallel no-CSIT (or multicast)
channel and by considering decentralization at the receivers.
%%%
The order optimality proof for the decentralized setting posed a number of technical challenges,
which were circumvented by involved mathematical manipulations and employing the notion of quasiconcavity.
%%%

%%%
The results in this paper can be extended in several interesting directions.
%%%
An intriguing direction would be to explore the fundamental limits of the considered
setup while relaxing the restriction of uncoded placement and one-shot linear delivery schemes.
%%%
While we expect uncoded placement to still be order optimal, the delivery scheme will likely rely
on interference interference alignment and symbol spreading.
%%%
This direction build upon and benefit from recent results reported in \cite{Xu2017,Hachem2018,Piovano2019}.
%%%
Another interesting direction would be to extend the setup and results in this paper to Fog-RAN architectures,
where decentralized placement can also be afforded at the transmitters due to the supporting cloud  \cite{Girgis2017,Xu2018}.
%%%
Such direction will also be relevant to D2D networks underlaying a cellular
infrastructure, that  performs the role of the cloud,
which can benefit from the lower complexity one-shot linear schemes.
%%%%%%%%%%%%%
%%%%%
\appendices
%%%
\section{Proof of Lemma \ref{theorem_outerbound}} \label{sec:outerbound}
Here we present the proof of Lemma \ref{theorem_outerbound}.
We start with the observation that under average distinct demands, as opposed to worst-case demands, there is a precise characterization for the number of packets to be delivered to the receivers \cite{Naderializadeh2017}.
%%%
Since the performance under average demands is no worse than that under worst-case demands,
the one-shot linear DoF in \eqref{eq:network_one_shot_linear_DoF} is bounded above by
\begin{equation}
\label{eq:upperbound_DoF_1}
\mathsf{DoF}_{\mathrm{L}}^*(\mu_{\mathrm{T}},\mu_{\mathrm{R}},\alpha) \leq  \frac{K_{\mathrm{R}}F(1-\mu_{\mathrm{R}})}{\overline{H}},
\end{equation}
where $\overline{H}$ is a lower bound on the delivery time under average demands rather than worst-case demands.
%%%
Note that the above relaxation is commonly used to obtain outer bounds in cache-aided setups, e.g. \cite{Maddah-Ali2014,Naderializadeh2017,Zhang2017,Piovano2019}.
%%%
Next, we follow the same general footsteps of \cite[Sec. V]{Naderializadeh2017} to
characterize and then find a lower bound for $\overline{H}$.
%%%
The steps borrowed from \cite{Naderializadeh2017} are explained in less detail, while we elaborate more on the new challenges that arise
due to packet splitting over the two subchannels.
\subsection{Upper bound on the Number of Subpackets Reliably Delivered Per Block}
First, let us fix the caching realization $\big( \{\mathcal{P}_i\}_{i=1}^{K_{\mathrm{T}}}, \{\mathcal{U}_i\}_{i=1}^{K_{\mathrm{R}}} \big)$,
user demand vector $\mathbf{d}$ and splitting ratio $q$.
As described in Section \ref{subsec:delivery_phase}, in each P-block or N-block, a subset of P-subpackets or N-subpacket are delivered over the
P-subchannel or the N-subchannel, respectively.
%%%
Let $\big\{\mathbf{w}_{n_l,f_l}^{(\mathrm{p})} \big\}_{l = 1}^{L^{(\mathrm{p})}}$  be a set of
$L^{(\mathrm{p})}$ P-subpackets to be delivered to $L^{(\mathrm{p})}$ distinct receivers over
one P-block, and $\big\{\mathbf{w}_{n_l,f_l}^{(\mathrm{n})} \big\}_{l = 1}^{L^{(\mathrm{n})}}$ be a set of
$L^{(\mathrm{n})}$ N-subpackets to be delivered to  $L^{(\mathrm{n})}$ distinct receivers over
one N-block.
%%%
In order for the receivers to successfully decode the transmitted subpackets,
$L^{(\mathrm{p})}$ and $L^{(\mathrm{n})}$ must satisfy
%%%
\begin{align}
\label{eq:condition_subpackets_p}
L^{(\mathrm{p})} & \leq \min_{l \in [L^{(\mathrm{p})}]} \big\{ |\mathcal{R}_l| +  |\mathcal{T}_l| \big\}  \\
\label{eq:condition_subpackets_n}
L^{(\mathrm{n})} & \leq \min_{l \in [L^{(\mathrm{n})}]} |\mathcal{R}_l| +  1
\end{align}
%%%
where, for any $l \in [L^{(\mathrm{p})}]$ or $l \in [L^{(\mathrm{n})}]$, $\mathcal{T}_l$ and $\mathcal{R}_l$ are the sets of transmitters and receivers, respectively,
which store the packet $\mathbf{w}_{n_l,f_l} = \big(\mathbf{w}_{n_l,f_l}^{(\mathrm{p})},\mathbf{w}_{n_l,f_l}^{(\mathrm{n})})$ in their caches.
%%%

%%%
The inequality in \eqref{eq:condition_subpackets_p} follows directly from  \cite[Lem. 3]{Naderializadeh2017}.
On the other hand, the inequality in \eqref{eq:condition_subpackets_n} can be
shown to hold by following the same general steps used to prove \cite[Lem. 3]{Naderializadeh2017}, while observing that the generic channel matrices and the lack of CSIT make the zero-forcing conditions in the proof of \cite[Lem. 3]{Naderializadeh2017}
impossible to satisfy almost surely.
This in turn eliminates the transmitter cooperation gain.
\subsection{Integer Program Formulation}
For any P-block and N-block indexed by $m^{(\mathrm{p})}$ and $m^{(\mathrm{n})}$ respectively,
the sets of subpackets $\mathcal{D}^{(\mathrm{p})}_{m^{(\mathrm{p})}}$ and $\mathcal{D}^{(\mathrm{n})}_{m^{(\mathrm{n})}}$  to be delivered
are deemed feasible \emph{only if} their cardinalities satisfy \eqref{eq:condition_subpackets_p} and \eqref{eq:condition_subpackets_n}.
Hence by keeping the caching realization, demand vector and splitting ratio fixed,
the following integer programming problems yields a lower bound on the
delivery time:
\begin{equation} \label{optimization_problem_P1}
\begin{aligned}
& {\min}
& & \max \Big \{ \frac{q}{\alpha} H^{(\mathrm{p})}, \frac{\bar{q}}{\bar{\alpha}} H^{(\mathrm{n})} \Big \}  \\
%%%
&\mathrm{s.t.} &&\bigcup_{m^{(\mathrm{p})}=1}^{H^{(\mathrm{p})}}{\mathcal{D}^{(\mathrm{p})}_{m^{(\mathrm{p})}} } = \bigcup_{r=1}^{K_{\mathrm{R}}}{\left(\mathcal{W}^{(\mathrm{p})}_{d_r} \setminus \mathcal{U}_r^{(\mathrm{p})} \right)} \\
&&&\bigcup_{m^{(\mathrm{n})}=1}^{H^{(\mathrm{n})}}{\mathcal{D}^{(\mathrm{n})}_{m^{(\mathrm{n})}} } = \bigcup_{r=1}^{K_{\mathrm{R}}}{\left(\mathcal{W}^{(\mathrm{n})}_{d_r} \setminus \mathcal{U}_r^{(\mathrm{n})} \right)} \\
&&& {\mathcal{D}^{(\mathrm{p})}_{m^{(\mathrm{p})}} }, {\mathcal{D}^{(\mathrm{n})}_{m^{(\mathrm{n})}} } \: \: \text{are feasible}, \: \: \forall m^{(\mathrm{p})} \in [H^{(\mathrm{p})}], \: \forall m^{(\mathrm{n})} \in [H^{(\mathrm{n})}].
\end{aligned}
\end{equation}
The optimal value for the above problem is denoted by $H^{*\left(\{\mathcal{P}_i\}_{i=1}^{K_{\mathrm{T}}}, \{\mathcal{U}_i\}_{i=1}^{K_{\mathrm{R}}}, \mathbf{d}, q \right)}$.
\subsection{From Worst-Case to Average Demands and Optimizing Over Caching Realizations and Splitting Ratios}
Given a caching realization $\big( {\{\mathcal{P}_i\}_{i=1}^{K_{\mathrm{T}}}, \{\mathcal{U}_i\}_{i=1}^{K_{\mathrm{R}}}} \big)$, each file $\mathcal{W}_{n}$, with $n \in [N]$, is split into
$(2^{K_{\mathrm{T}}}-1)( 2^{K_{\mathrm{R}}})$ subfiles $\{\mathcal{W}_{n,\mathcal{T},\mathcal{R}}\}_{ \mathcal{T} \subseteq_{\emptyset}[K_{\mathrm{T}}], \mathcal{R} \subseteq [K_{\mathrm{R} }]}$, where $\mathcal{W}_{n,\mathcal{T},\mathcal{R}}$ denotes the subfile of file $\mathcal{W}_{n}$ cached by transmitters in $\mathcal{T}$ and receivers in $\mathcal{R}$, and
$\mathcal{T} \subseteq_{\emptyset} [K_{\mathrm{T}}] $ denotes $\mathcal{T} \subseteq [K_{\mathrm{T}}], \mathcal{T} \neq \emptyset$.
%%%
Denoting the number of packets in ${\mathcal{W}_{n, \mathcal{T}, \mathcal{R}}}$ as ${a_{n, \mathcal{T}, \mathcal{R}}}$, we may write an optimization problem to minimize $H^{*\left(\{\mathcal{P}_i\}_{i=1}^{K_{\mathrm{T}}}, \{\mathcal{U}_i\}_{i=1}^{K_{\mathrm{R}}}, \mathbf{d}, q \right)}$, for the worst-case demands, over all caching realizations and splitting ratios.
%%

%%%
As in \cite{Naderializadeh2017}, we further lower bound the delivery time by considering average demands instead of worst-case demands.
%%%
In particular, by taking the average over the set of all possible {$\pi(N,K_{\mathrm{R}}) = \frac{N!}{(N-K_{\mathrm{R}})!}$} permutations of distinct receiver demands, denote by $\mathcal{P}_{N,K_{\mathrm{R}}}$, we write the problem:
%%%%
\begin{equation}
\label{optim_average_demands}
\begin{aligned}
& {\min_{\{\mathcal{P}_i\}_{i=1}^{K_{\mathrm{T}}}, \{\mathcal{U}_i\}_{i=1}^{K_{\mathrm{R}}},q}} \:
& & \frac{1}{\pi(N,K_{\mathrm{R}})}\sum_{\mathbf{d} \in \mathcal{P}_{N,K_{\mathrm{R}}} }{H^{*\left(\{\mathcal{P}_i\}_{i=1}^{K_{\mathrm{T}}}, \{\mathcal{U}_i\}_{i=1}^{K_{\mathrm{R}}} , \mathbf{d}, q \right)}} \\
& \mathrm{s.t.}
& &  \sum_{\substack{\mathcal{T} \subseteq_{\emptyset}  [K_{\mathrm{T}}] }}\sum_{\substack{\mathcal{R} \subseteq  [K_{\mathrm{R}}]}} {a_{n, \mathcal{T}, \mathcal{R}}}  = F, \: \forall n \in [N]\\
&&& \sum_{n=1}^N \sum_{\substack{\mathcal{T} \subseteq  [K_{\mathrm{T}}]: \\ i \in \mathcal{T} }}\sum_{\substack{\mathcal{R} \subseteq  [K_{\mathrm{R}}]}}  {a_{n, \mathcal{T}, \mathcal{R}}} \leq   \mu_{\mathrm{T}} NF, \: \forall i \in [K_{\mathrm{T}}]\\
&&& \sum_{n=1}^N \sum_{\substack{\mathcal{T} \subseteq_{\emptyset}  [K_{\mathrm{T}}] }}\sum_{\substack{\mathcal{R} \subseteq  [K_{\mathrm{R}}] : \\ j \in \mathcal{R}}}  {a_{n, \mathcal{T}, \mathcal{R}}}  \leq \mu_{\mathrm{R}}NF, \: \forall j \in [K_{\mathrm{R}}] \\
&&& q \in [0,1], a_{n,\mathcal{T},\mathcal{R}} \geq 0, \forall n \in [N],     \forall \mathcal{T} \subseteq_{\emptyset} [K_{\mathrm{T}}], \forall \mathcal{R} \subseteq [K_{\mathrm{R}}].
\end{aligned}
\end{equation}
%%%
The optimum objective for the above problem is denoted by $\overline{H}$, which appears in the bound in \eqref{eq:upperbound_DoF_1}.
%%%
In what follows, we are interested in further lower bounding $\overline{H}$.
%%%
%%
\subsection{Decoupling the P and N Subchannel and Optimizing Over Caching Realizations}
\label{sec:P_N_subchannel}
To obtain a lower bound for $\bar{H}$, we consider optimizing over caching realizations for the  P-subchannel and N-subchannel independently.
To facilitate this, we start by observing that  $H^{*\left(\{\mathcal{P}_i\}_{i=1}^{K_{\mathrm{T}}}, \{\mathcal{U}_i\}_{i=1}^{K_{\mathrm{R}}} , \mathbf{d}, q \right)}$
in \eqref{optim_average_demands}, the optimum objective of \eqref{optimization_problem_P1} is bounded below as
%%%
\begin{equation}
\label{eq:H_given_cache_lower_bound}
H^{*\left(\{\mathcal{P}_i\}_{i=1}^{K_{\mathrm{T}}}, \{\mathcal{U}_i\}_{i=1}^{K_{\mathrm{R}}} , \mathbf{d}, q \right)}
\geq \max \Big \{ \frac{q}{\alpha} {H^{(\mathrm{p})*}}^{\left(\{\mathcal{P}_i\}_{i=1}^{K_{\mathrm{T}}}, \{\mathcal{U}_i\}_{i=1}^{K_{\mathrm{R}}}, \mathbf{d}, q \right)}
, \frac{\bar{q}}{\bar{\alpha}}
{H^{(\mathrm{n})*}}^{\left(\{\mathcal{P}_i\}_{i=1}^{K_{\mathrm{T}}}, \{\mathcal{U}_i\}_{i=1}^{K_{\mathrm{R}}}, \mathbf{d}, q \right)} \Big \}
\end{equation}
%%%

where ${H^{(\mathrm{s})*}}^{\left(\{\mathcal{P}_i\}_{i=1}^{K_{\mathrm{T}}}, \{\mathcal{U}_i\}_{i=1}^{K_{\mathrm{R}}}, \mathbf{d}, q \right)}$,
$\mathrm{s} \in \{\mathrm{p},\mathrm{n}\}$, is the optimum objective of the optimization problem
%%%
\begin{equation}
\begin{aligned}
\label{eq:H_s_given_cache}
& {\min}
& & H^{(\mathrm{s})}  \\
&\mathrm{s.t.}&&\bigcup_{m^{(\mathrm{s})}=1}^{H^{(\mathrm{s})}}{\mathcal{D}^{(\mathrm{s})}_{m^{(\mathrm{s})}} } = \bigcup_{r=1}^{K_{\mathrm{r}}}{\left(\mathcal{W}^{(\mathrm{s})}_{d_r} \setminus \mathcal{U}_r^{(\mathrm{s})} \right)} \\
&&& {\mathcal{D}^{(\mathrm{s})}_{m^{(\mathrm{s})}} } \: \: \text{is feasible}, \: \: \forall m^{(\mathrm{s})} \in [H^{(\mathrm{s})}].
\end{aligned}
\end{equation}
The lower bound in \eqref{eq:H_given_cache_lower_bound} is derived directly from problem \eqref{optimization_problem_P1},
e.g. the P-subchannel term on the right-hand side of \eqref{eq:H_given_cache_lower_bound} is obtained by
relaxing all N-subchannel components in the objective and constraints of problem \eqref{optimization_problem_P1}.
%%%%
Denoting the average demand operator $\frac{1}{\pi(N,K_{\mathrm{R}})}\sum_{\mathbf{d} \in \mathcal{P}_{N,K_{\mathrm{R}}} }(\cdot)$
by $\E_{\mathbf{d}}(\cdot)$ for brevity,
it follows that the objective function of problem \eqref{optim_average_demands} is lower bounded as
%%%
\begin{align}
\nonumber
\E_{\mathbf{d}}  \Big( &{H^{*\left(\{\mathcal{P}_i\}_{i=1}^{K_{\mathrm{T}}}, \{\mathcal{U}_i\}_{i=1}^{K_{\mathrm{R}}} , \mathbf{d}, q \right)}} \Big) \geq
\E_{\mathbf{d}} \Big( \max \Big \{ \frac{q}{\alpha} {H^{(\mathrm{p})*}}^{\left(\{\mathcal{P}_i\}_{i=1}^{K_{\mathrm{T}}}, \{\mathcal{U}_i\}_{i=1}^{K_{\mathrm{R}}}, \mathbf{d}, q \right)}
, \frac{\bar{q}}{\bar{\alpha}} {H^{(\mathrm{n})*}}^{\left(\{\mathcal{P}_i\}_{i=1}^{K_{\mathrm{T}}}, \{\mathcal{U}_i\}_{i=1}^{K_{\mathrm{R}}}, \mathbf{d}, q \right)} \Big \} \Big)
\\
\label{eq:H_given_cache_lower_bound_2}
& \geq
\max \Big \{\frac{q}{\alpha} \E_{\mathbf{d}} \Big( {H^{(\mathrm{p})*}}^{\left(\{\mathcal{P}_i\}_{i=1}^{K_{\mathrm{T}}}, \{\mathcal{U}_i\}_{i=1}^{K_{\mathrm{R}}}, \mathbf{d}, q \right)} \Big)
, \frac{\bar{q}}{\bar{\alpha}} \E_{\mathbf{d}} \Big(
{H^{(\mathrm{n})*}}^{\left(\{\mathcal{P}_i\}_{i=1}^{K_{\mathrm{T}}}, \{\mathcal{U}_i\}_{i=1}^{K_{\mathrm{R}}}, \mathbf{d}, q \right)} \Big) \Big \}
\end{align}
%%%
where the inequality in \eqref{eq:H_given_cache_lower_bound_2} follows from the convexity of the pointwise maximum function and Jensen's inequality. %%%
Next, we plug the lower bound in \eqref{eq:H_given_cache_lower_bound_2} into \eqref{optim_average_demands} from which we obtain a lower bound on $\overline{H}$.
%%%
Moreover, for any given splitting ratio $q$, we optimize over caching realizations independently for the P-subchannel and N-subchannel through
%%%
\begin{equation}
\label{eq:optim_prob_Hb}
\begin{aligned}
& {\min_{\{\mathcal{P}_i\}_{i=1}^{K_{\mathrm{T}}}, \{\mathcal{U}_i\}_{i=1}^{K_{\mathrm{R}}}}} \:
& &  \frac{1}{\pi(N,K_{\mathrm{R}})}\sum_{\mathbf{d} \in \mathcal{P}_{N,K_{\mathrm{R}}} }{{H^{(\mathrm{s})*}}^{\left(\{\mathcal{P}_i\}_{i=1}^{K_{\mathrm{T}}}, \{\mathcal{U}_i\}_{i=1}^{K_{\mathrm{R}}}, \mathbf{d}, q \right)}}\\
& \mathrm{s.t.}
& &  \sum_{\substack{\mathcal{T} \subseteq_{\emptyset}  [K_{\mathrm{T}}] }}\sum_{\substack{\mathcal{R} \subseteq  [K_{\mathrm{R}}]}} {a_{n, \mathcal{T}, \mathcal{R}}}  = F, \: \forall n \in [N]\\
&&& \sum_{n=1}^N \sum_{\substack{\mathcal{T} \subseteq  [K_{\mathrm{T}}]: \\ i \in \mathcal{T} }}\sum_{\substack{\mathcal{R} \subseteq  [K_{\mathrm{R}}]}}  {a_{n, \mathcal{T}, \mathcal{R}}}  \leq \mu_{\mathrm{T}}NF, \: \forall i \in [K_{\mathrm{T}}]\\
&&& \sum_{n=1}^N \sum_{\substack{\mathcal{T} \subseteq_{\emptyset}  [K_{\mathrm{T}}] }}\sum_{\substack{\mathcal{R} \subseteq  [K_{\mathrm{R}}] : \\ j \in \mathcal{R}}}  {a_{n, \mathcal{T}, \mathcal{R}}}  \leq \mu_{\mathrm{R}}NF, \: \forall j \in [K_{\mathrm{R}}] \\
&&& a_{n,\mathcal{T},\mathcal{R}} \geq 0, \forall n \in [N],     \forall \mathcal{T} \subseteq_{\emptyset} [K_{\mathrm{T}}], \forall \mathcal{R} \subseteq [K_{\mathrm{R}}],
\end{aligned}
\end{equation}
%%%
for which we denote the optimum objective function as $\overline{H^{(\mathrm{s})}}^{ \left(q \right)}$, $\mathrm{s} \in \{\mathrm{p},\mathrm{n}\}$.
%%%%
This yields the lower bound on $\overline{H}$ given by
%%%
\begin{equation}
\label{eq:average_ndt}
\overline{H} \geq \min_{q \in [0,1]}  \max \Big \{ \frac{q}{\alpha} \overline{H^{(\mathrm{p})}} ^{\left( q \right)}, \frac{\bar{q}}{\bar{\alpha}} \overline{H^{(\mathrm{n})}}^{\left( q \right)} \Big \}.
\end{equation}
%%%
The two components $\overline{H^{(\mathrm{p})}}^{ \left(q \right)}$ and
$\overline{H^{(\mathrm{n})}}^{ \left(q \right)}$ can be separately lower bounded as
%%%
%%
\begin{align}
\label{eq:H_p lowerbound}
\overline{H^{(\mathrm{p})}}^{(q)} & \geq  \frac{ K_{\mathrm{R}}F(1-\mu_{\mathrm{R}})^2}{ {K_{\mathrm{T}}\mu_{\mathrm{T}}+K_{\mathrm{R}}\mu_{\mathrm{R}}}} \\
%%%
\label{eq:H_n lowerbound}
\overline{H^{(\mathrm{n})}}^{(q)} & \geq  \frac{ K_{\mathrm{R}}F(1-\mu_{\mathrm{R}})^2}{ {1+K_{\mathrm{R}}\mu_{\mathrm{R}}}}.
\end{align}
The lower bound in \eqref{eq:H_p lowerbound} follows directly from \cite[Lem. 4]{Naderializadeh2017}.
%%%
On the other hand, the lower bound in \eqref{eq:H_n lowerbound} is derived at the end of this section
by employing the same techniques in the proof of \cite{Naderializadeh2017}.
%%%

%%%
Since in the problem in (\ref{eq:optim_prob_Hb}) the total
number of subpackets per block delivered over either of the two subchannels is
$K_{\mathrm{R}}F\left(1-\mu_{\mathrm{R}} \right)$, and no
more than $K_{\mathrm{R}}$ subpackets can be delivered simultaneously, we obtain
$\overline{H^{(\mathrm{s})}}^{(q)} \geq \frac{ K_{\mathrm{R}} F(1-\mu_{\mathrm{R}})}{ K_{\mathrm{R}} }$.
%%%
Combining this with the lower bounds in \eqref{eq:H_p lowerbound} and \eqref{eq:H_n lowerbound}, we obtain
\begin{align}
\label{eq:lb_H_p}
\overline{H^{(\mathrm{p})}}^{(q)} & \geq \frac{ K_{\mathrm{R}}F(1-\mu_{\mathrm{R}})}{ \min \left\{ \frac{K_{\mathrm{T}}\mu_{\mathrm{T}}+K_{\mathrm{R}}\mu_{\mathrm{R}}}{1-\mu_{\mathrm{R}}} , K_{\mathrm{R}}          \right\}} \\
\label{eq:lb_H_n}
\overline{H^{(\mathrm{n})}}^{(q)} & \geq \frac{ K_{\mathrm{R}}F(1-\mu_{\mathrm{R}})}{\min \left\{ \frac{1+K_{\mathrm{R}}\mu_{\mathrm{R}}}{1-\mu_{\mathrm{R}}} , K_{\mathrm{R}}          \right\}}.
\end{align}
%%%
It is evident that the above lower bounds do not depend on the value of $q$, and by combining (\ref{eq:lb_H_p}) and (\ref{eq:lb_H_n}) with (\ref{eq:average_ndt}), it follows that
\begin{equation}
\label{eq:H_overline_lower_bound}
\overline{H} \geq \min_{q \in [0,1]} \max \Bigg \{ \frac{q}{\alpha} \cdot \frac{ K_{\mathrm{R}}F(1-\mu_{\mathrm{R}})}{\min \Big \{ \frac{K_{\mathrm{T}}\mu_{\mathrm{T}}+K_{\mathrm{R}}\mu_{\mathrm{R}}}{1-\mu_{\mathrm{R}}} , K_{\mathrm{R}}          \Big \}}       , \frac{\bar{q}}{\bar{\alpha}} \cdot \frac{ K_{\mathrm{R}}F(1-\mu_{\mathrm{R}})}{ \min \Big \{ \frac{1+K_{\mathrm{R}}\mu_{\mathrm{R}}}{1-\mu_{\mathrm{R}}} , K_{\mathrm{R}}          \Big \}}    \Bigg \}.
\end{equation}
%%
%%%%%%%%%%%%%%%%%%
\subsection{Optimizing Over Splitting Rations and Combing Bounds}
%%%%%%%%%%%%%%%%%%
%%
The splitting ration $q$ that minimizes the right-hand side of \eqref{eq:H_overline_lower_bound},
which we denote by $q^*$, must satisfy
\begin{equation}
\nonumber
\frac{q^*}{\alpha} \cdot \frac{ K_{\mathrm{R}}F(1-\mu_{\mathrm{R}})}{ \min \Big \{ \frac{K_{\mathrm{T}}\mu_{\mathrm{T}}+K_{\mathrm{R}}\mu_{\mathrm{R}}}{1-\mu_{\mathrm{R}}} , K_{\mathrm{R}}          \Big \}}       = \frac{\bar{q^*}}{\bar{\alpha}} \cdot \frac{ K_{\mathrm{R}}F(1-\mu_{\mathrm{R}})}{ \min \Big \{ \frac{1+K_{\mathrm{R}}\mu_{\mathrm{R}}}{1-\mu_{\mathrm{R}}} , K_{\mathrm{R}}          \Big \}},
\end{equation}
as any other $q$ leads to a larger value for the right-hand side of \eqref{eq:H_overline_lower_bound}.
%%%
By considering $q^*$, we obtain\footnote{For any real numbers $x,y$ and $q$ such that $\frac{q}{x} = \frac{1-q}{y}$, it is easy to verify that $\frac{q}{x}  = \frac{1}{x+y}$.}
\begin{equation}
\label{lower_bound_NDT}
\bar{H}  \geq \frac{K_{\mathrm{R}}F(1-\mu_{\mathrm{R}})}{ \alpha \cdot \min \Big \{ \frac{K_{\mathrm{T}}\mu_{\mathrm{T}}+K_{\mathrm{R}}\mu_{\mathrm{R}}}{1-\mu_{\mathrm{R}}} , K_{\mathrm{R}}\Big \} + \bar{\alpha} \cdot \min \Big \{\frac{1+K_{\mathrm{R}}\mu_{\mathrm{R}}}{1-\mu_{\mathrm{R}}} , K_{\mathrm{R}}\Big \}}.
\end{equation}
Combining the lower bound in \eqref{lower_bound_NDT} with the upper bound in \eqref{eq:upperbound_DoF_1}, we obtain
\begin{equation}
\nonumber
\mathsf{DoF}_{\mathrm{L}}^*(\mu_{\mathrm{T}},\mu_{\mathrm{R}},\alpha) \leq  \alpha \cdot \min \Big \{ \frac{K_{\mathrm{T}}\mu_{\mathrm{T}}+K_{\mathrm{R}}\mu_{\mathrm{R}}}{1-\mu_{\mathrm{R}}}, K_{\mathrm{R} }
\Big \} + \bar{\alpha} \cdot \min \Big \{\frac{1+K_{\mathrm{R}}\mu_{\mathrm{R}}}{1-\mu_{\mathrm{R}}},K_{\mathrm{R}}\Big \}
\end{equation}
which concludes the proof of Lemma \ref{theorem_outerbound}.
\subsection{Proof of the lower bound in \eqref{eq:H_n lowerbound}}
Note that $\overline{H^{(\mathrm{n})}}^{(q)}$ corresponds to the optimum objective value for the optimization problem in (\ref{eq:optim_prob_Hb})
when $\mathrm{s} = \mathrm{n}$.
%%%
To bound this, we follow here the footsteps in the proof of  \cite[Lem. 4]{Naderializadeh2017}.
Starting from ${H^{(\mathrm{n})*}}^{\left(\{\mathcal{P}_i\}_{i=1}^{K_{\mathrm{T}}}, \{\mathcal{U}_i\}_{i=1}^{K_{\mathrm{R}}}, \mathbf{d}, q \right)}$ and
by invoking \eqref{eq:condition_subpackets_n}, we obtain
\begin{equation}
{H^{(\mathrm{n})*}}^{\left(\{\mathcal{P}_i\}_{i=1}^{K_{\mathrm{T}}}, \{\mathcal{U}_i\}_{i=1}^{K_{\mathrm{R}}} ,q, \mathbf{d} \right)}
\geq \sum_{i=1}^{K_{\mathrm{T}}} \sum_{j=0}^{K_{\mathrm{R}}} \sum_{r=1}^{K_{\mathrm{R}}} \sum_{\substack{\mathcal{T} \subseteq  [K_{\mathrm{T}}]: \\ |\mathcal{T}|=i}}\sum_{\substack{\mathcal{R} \subseteq  [K_{\mathrm{R}}]: \\ |\mathcal{R}|=j \\ r \notin \mathcal{R}}} \frac{a_{d_r, \mathcal{T}, \mathcal{R}}}{j+1}.
\end{equation}
By averaging over all possible demands, we obtain
\begin{equation}
\begin{split}
\overline{H^{(\mathrm{n})}}^{ \left(\{P\}_{i=1}^{\mathrm{K_{\mathrm{T}}}}, \{\mathcal{U}_i\}_{i=1}^{\mathrm{K_{\mathcal{T}}}}, q \right)} & \geq \frac{1}{\pi(N,K_{\mathrm{R}})}\sum_{i=1}^{K_{\mathrm{T}}} \sum_{j=0}^{K_{\mathrm{R}}} \sum_{r=1}^{K_{\mathrm{R}}} \sum_{\substack{\mathcal{T} \subseteq  [K_{\mathrm{T}}]: \\ |\mathcal{T}|=i}}\sum_{\substack{\mathcal{R} \subseteq  [K_{\mathrm{R}}]: \\ |\mathcal{R}|=j \\ r \notin \mathcal{R}}} {\pi(N-1,K_{\mathrm{R}}-1)} \sum_{n=1}^N\frac{a_{n, \mathcal{T}, \mathcal{R}}}{ j +1} \\
& = \frac{1}{N}\sum_{i=1}^{K_{\mathrm{T}}} \sum_{j=0}^{K_{\mathrm{R}}-1} \frac{w_{i,j}}{ j +1} .
\end{split}
\end{equation}
where, for any $i \in [K_{\mathrm{T}}]$ and $j \in [K_{\mathrm{R}}-1] \cup {0}$, we define
\begin{equation}
w_{i,j}=\sum_{r=1}^{K_{\mathrm{R}}} \sum_{\substack{\mathcal{T} \subseteq  [K_{\mathrm{T}}]: \\ |\mathcal{T}|=i}}\sum_{\substack{\mathcal{R} \subseteq  [K_{\mathrm{R}}]: \\ |\mathcal{R}|=j \\ r \notin \mathcal{R}}} \sum_{n=1}^N {a_{n, \mathcal{T}, \mathcal{R}}}=(K_{\mathrm{R}}-j) \sum_{\substack{\mathcal{T} \subseteq  [K_{\mathrm{T}}]: \\ |\mathcal{T}|=i}}\sum_{\substack{\mathcal{R} \subseteq  [K_{\mathrm{R}}]: \\ |\mathcal{R}|=j}} \sum_{n=1}^N {a_{n, \mathcal{T}, \mathcal{R}}}.
\end{equation}
It is readily seen that
\begin{equation} \label{eq:receiver_size_app}
K_{\mathrm{R}}\mu_{\mathrm{R}}NF  \geq \sum_{r=1}^{K_{\mathrm{R}}} \sum_{i=1}^{K_{\mathrm{T}}} \sum_{j=0}^{K_{\mathrm{R}}} \sum_{\substack{\mathcal{T} \subseteq  [K_{\mathrm{T}}]: \\ |\mathcal{T}|=i  }}\sum_{\substack{\mathcal{R} \subseteq  [K_{\mathrm{R}}]: \\ |\mathcal{R}|=j \\ r \in \mathcal{R} }} \sum_{n=1}^N {a_{n, \mathcal{T}, \mathcal{R}}}
\geq \sum_{i=1}^{K_{\mathrm{T}}} \sum_{j=0}^{K_{\mathrm{R}}-1}  \frac{j}{K_{\mathrm{R}}-j} w_{i,j}
\end{equation}
and
\begin{equation} \label{eq:library_size_app}
NF  =  \sum_{i=1}^{K_{\mathrm{T}}} \sum_{j=0}^{K_{\mathrm{R}}} \sum_{\substack{\mathcal{T} \subseteq  [K_{\mathrm{T}}]: \\ |\mathcal{T}|=i  }}\sum_{\substack{\mathcal{R} \subseteq  [K_{\mathrm{R}}]: \\ |\mathcal{R}|=j }} \sum_{n=1}^N {a_{n, \mathcal{T}, \mathcal{R}}} \\
\geq \sum_{i=1}^{K_{\mathrm{T}}} \sum_{j=0}^{K_{\mathrm{R}}-1}  \frac{1}{K_{\mathrm{R}}-j} w_{i,j}.
\end{equation}
After applying the Cauchy-Schwarz inequality, we obtain
\begin{equation}
\sum_{j=0}^{K_{\mathrm{R}}-1}{w_{i,j}} \leq \sqrt{\sum_{j=0}^{K_{\mathrm{R}}-1}\frac{j+1}{K_{\mathrm{R}}-j} w_{i,j}} \cdot  \sqrt{\sum_{j=0}^{K_{\mathrm{R}}-1}\frac{K_{\mathrm{R}}-j}{j+1} w_{i,j}}.
\end{equation}
Applying the Cauchy-Schwarz inequality again, we obtain
\begin{equation}
\sum_{i=1}^{K_{\mathrm{T}}}\sum_{j=0}^{K_{\mathrm{R}}-1}{w_{i,j}}
\leq \sqrt{ \sum_{i=1}^{K_\mathrm{T}} \sum_{j=0}^{K_{\mathrm{R}}-1}\frac{j+1}{K_{\mathrm{R}}-j} w_{i,j}} \cdot  \sqrt{\sum_{i=1}^{K_\mathrm{T}} \sum_{j=0}^{K_{\mathrm{R}}-1}\frac{K_{\mathrm{R}}-j}{j+1} w_{i,j}}.
\end{equation}
Moreover, from (\ref{eq:receiver_size_app}) and (\ref{eq:library_size_app}) we know that
\begin{equation}
\sum_{i=1}^{K_\mathrm{T}} \sum_{j=0}^{K_{\mathrm{R}}-1}\frac{j+1}{K_{\mathrm{R}}-j} w_{i,j}
\leq K_{\mathrm{R}}\mu_{\mathrm{R}}NF +  NF.
\end{equation}
It follows that
\begin{equation}
\sum_{i=1}^{K_{\mathrm{T}}}\sum_{j=0}^{K_{\mathrm{R}}-1}{w_{i,j}} \leq
\sqrt{ K_{\mathrm{R}}\mu_{\mathrm{R}}NF +  NF } \cdot  \sqrt{\sum_{i=1}^{K_\mathrm{T}} \sum_{j=0}^{K_{\mathrm{R}}-1}\frac{K_{\mathrm{R}}-j}{j+1} w_{i,j}}.
\end{equation}
Furthermore, from \cite{Naderializadeh2017} we know that
\begin{equation}
\sum_{i=1}^{K_{\mathrm{T}}}\sum_{j=0}^{K_{\mathrm{R}}-1}{w_{i,j}}
\geq K_{\mathrm{R}}N(1-\mu_{\mathrm{R}})F.
\end{equation}
Hence, it follows that
\begin{equation}
\begin{split}
\overline{H^{(\mathrm{n})}}^ {\left(\{\mathcal{P}_i\}_{i=1}^{\mathrm{K_{\mathrm{T}}}}, \{\mathcal{	U}_i\}_{i=1}^{\mathrm{K_{\mathcal{T}}}} ,   q\right)}  & \geq  \frac{1}{N}\sum_{i=1}^{K_{\mathrm{T}}} \sum_{j=0}^{K_{\mathrm{R}}-1} \frac{w_{i,j}}{ j +1} \geq \frac{1}{K_{\mathrm{R}}N}\sum_{i=1}^{K_{\mathrm{T}}} \sum_{j=0}^{K_{\mathrm{R}}-1} \frac{K_{\mathrm{R}}-j }{j +1} w_{i,j} \\
& \geq \frac{1}{K_{\mathrm{R}}N} \cdot \frac{1}{ K_{\mathrm{R}}\mu_{\mathrm{R}}NF +NF} \left( \sum_{i=1}^{K_{\mathrm{T}}}\sum_{j=0}^{K_{\mathrm{R}}-1}{w_{i,j}} \right)^2 \\
& \geq \frac{K_{\mathrm{R}}NF \left( 1 - \mu_{\mathrm{R}} \right)^2}{ K_{\mathrm{R}}\mu_{\mathrm{R}}N + N}= \frac{K_{\mathrm{R}}F \left( 1 - \mu_{\mathrm{R}} \right)^2}{1+K_{\mathrm{R}}\mu_{\mathrm{R}}}
\end{split}
\end{equation}
for any caching realization $\big(\{\mathcal{P}_i\}_{i=1}^{\mathrm{K_{\mathrm{T}}}}, \{\mathcal{	U}_i\}_{i=1}^{\mathrm{K_{\mathcal{T}}}} \big)$, which concludes the proof.

\section{Proof of Lemma \ref{lemma:poly_con_decen}}
\label{sec:proof_poly_property}
Here we present a proof of the inequality in \eqref{eq:poly_con_decen}.
%%%
We start with the following instrumental lemma.
%%%
\begin{lemma} \label{lemma_appendix_D}
Consider a polynomial $\phi(\zeta)=\sum_{m=0}^d {a_m \zeta^m}$ for which there exists and integer $N$ in $[-1,d]_{\mathbb{Z}}$
such that the coefficients of $\phi(\zeta)$ satisfy the following condition
\begin{equation}
\label{eq:cond_lemma_appendix}
a_m  \geq 0, \ m < N
\quad \text{and} \quad
a_m  > 0, \ m = N
\quad \text{and} \quad
a_m  \leq 0, \ m > N
%%%
\end{equation}
where the case $N = -1$ implies $a_0,\ldots,a_{d} \leq 0$.
The polynomial $\phi(\zeta)$ is quasiconcave on $\zeta \in [0,\infty)$.
\end{lemma}
\begin{proof}
First, we note that for the cases: $N=-1$ (i.e. when $a_m \leq 0$ for all $m$), $N=0$ and $N=1$,
the second derivative of $\phi(\zeta)$ is a polynomial with all coefficients not greater than zero.
Therefore, $\phi(\zeta)$ is concave, and hence quasiconcave, on $ \zeta \in [0,\infty)$.
We proceed by induction.
%%%
In particular, assume that the quasiconcavity  hypothesis holds for all polynomials the satisfy the condition
in \eqref{eq:cond_lemma_appendix} for integer $N=n$, where $n \geq 1$.
%%%
Now consider a polynomial $\phi(\zeta)$ that satisfies the condition  in \eqref{eq:cond_lemma_appendix} for $N=n+1$.
%%%
It is readily seen that the first derivative of $\phi(\zeta)$, denoted by
$\phi'(\zeta)$, is a polynomial which satisfies the condition in \eqref{eq:cond_lemma_appendix} for $N=n$.
%%%
Hence, $\phi'(\zeta)$ is quasiconcave by the induction hypothesis.
Moreover, as $n \geq 1$, it follows from  \eqref{eq:cond_lemma_appendix} that  $\phi'(0) \geq 0$.
It can be verified that $\phi'(0) \geq 0$ combined with the quasiconcavity of $\phi'(\zeta)$ guarantee that: either $\phi'(\zeta)$
is non-negative over $[0, \infty)$, or there exists $\zeta' \in [0, \infty)$ such that $\phi'(\zeta) \geq 0$ over the interval $[0,\zeta']$ and $\phi'(\zeta) \leq 0 $ over the interval $[\zeta', \infty)$.
It follows that $\phi(\zeta)$ is eithrer non-decreasing over $[0, \infty)$, or non-decreasing over $[0,\zeta']$ and non-increasing over $[\zeta', \infty)$.
In both cases, $\phi(\zeta)$ is quasiconcave. This concludes the proof of Lemma \ref{lemma_appendix_D}.
\end{proof}
%%%%
Next, we show that the coefficients of the polynomial $p(\zeta)$ of interest satisfy the
conditions in Lemma \ref{lemma_appendix_D}.
%%%
As this shows that $p(\zeta)$  is quasiconcave, the inequality in  \eqref{eq:poly_con_decen} directly follows by definition.
%%%
The remainder of this appendix is dedicated to showing that $p(\zeta)$  is an instance of Lemma \ref{lemma_appendix_D}.
%%%

The key step of this proof is to show that the sequence $\Big\{ \frac{c_m}{\binom{K_{\mathrm{R}}-1}{m-1}} \Big\}_{m = 1}^{K_{\mathrm{R}}-1}$
is non-increasing.
Supposing that this holds true, then this sequence would satisfy the condition of Lemma \ref{lemma_appendix_D}, applied only to the indices $ m \in [1,  K_{\mathrm{R}}-1 ]_{\mathbb{Z}}$.
Since the sign of $\frac{c_m}{\binom{K_{\mathrm{R}}-1}{m-1}} $
is preserved by $c_m$, then $\{c_m\}_{m = 1}^{K_{\mathrm{R}}-1}$ also satisfies
the condition of Lemma \ref{lemma_appendix_D} over $ m \in [1,  K_{\mathrm{R}}-1 ]_{\mathbb{Z}}$.
%%%
Combining this with  $c_0 = 0$ and $c_{K_{\mathrm{R}}} \leq 0$, it follows that
$\{c_m\}_{m = 0}^{K_{\mathrm{R}}}$ satisfies the condition of Lemma \ref{lemma_appendix_D}, which in turn concludes the proof.
%%%
Therefore, our problem reduces to showing that $\frac{c_m}{\binom{K_{\mathrm{R}}-1}{m-1}} $ in a non-increasing over
$ m \in [1,  K_{\mathrm{R}}-1 ]_{\mathbb{Z}}$.
%%%

%%%
First, it is readily seen that $c_m$ can be written as
 \begin{equation}
 \nonumber
 c_m = \binom{K_{\mathrm{R}}-1}{m-1}  \left[ \left( \frac{K_{\mathrm{R}}+1}{m} - \frac{K_{\mathrm{R}}+r}{\min \{ r+m -1 ,  K_{\mathrm{R}} \}} \right)  + \frac{K_{\mathrm{R}}-m}{m} \left( \frac{1}{m+1} - \frac{r}{\min \{ r+m ,  K_{\mathrm{R}} \}} \right) \right].
 \end{equation}
For briefness, we denote the coefficient $\frac{c_m}{\binom{K_{\mathrm{R}}-1}{m-1}}$ as $c'_m$.
Hence, $c'_m$ is given by
\begin{equation}
\nonumber
c'_m= \left( \frac{K_{\mathrm{R}}+1}{m} - \frac{K_{\mathrm{R}}+r}{\min \{ r+m -1 ,  K_{\mathrm{R}} \}} \right)  + \frac{K_{\mathrm{R}}-m}{m} \left( \frac{1}{m+1} - \frac{r}{\min \{ r+m ,  K_{\mathrm{R}} \}} \right).
\end{equation}
%%%
Next, let us define the integer $\tilde{r} \in [1,K_{\mathrm{R}}-1]_{\mathbb{Z}}$ as
$\tilde{r}\triangleq \lfloor r \rfloor = r - \epsilon$, where $\epsilon \in [0,1)$.
%%%
Using this definition, it can be shown that $c'_m$, $ m \in [1,  K_{\mathrm{R}}-1 ]_{\mathbb{Z}}$,  may be expressed as:
\begin{equation}
\nonumber
c'_m =
\begin{cases}
d_m \triangleq
\left( \frac{K_{\mathrm{R}}+1}{m} - \frac{K_{\mathrm{R}}+r}{  r+m -1 } \right)  + \frac{K_{\mathrm{R}}-m}{m} \left( \frac{1}{m+1} - \frac{r}{ r+m} \right),
\  m \in [1,K_{\mathrm{R}}-\tilde{r} -1 ]_{\mathbb{Z}} \\
%%%%
\left( \frac{K_{\mathrm{R}}+1}{K_{\mathrm{R}}-\tilde{r}} - \frac{K_{\mathrm{R}}+\tilde{r} + \epsilon}{ K_{\mathrm{R}}+\epsilon -1 } \right)  + \frac{\tilde{r}}{K_{\mathrm{R}}-\tilde{r}}  \left( \frac{1}{K_{\mathrm{R}}-\tilde{r}+1} - \frac{\tilde{r} + \epsilon}{K_{\mathrm{R}}} \right) , \ m = K_{\mathrm{R}}-\tilde{r}\\
%%%%
e_m \triangleq \left( \frac{K_{\mathrm{R}}+1}{m} - \frac{K_{\mathrm{R}}+r}{ K_{\mathrm{R}} } \right)  + \frac{K_{\mathrm{R}}-m}{m}\left( \frac{1}{m+1} - \frac{r}{K_{\mathrm{R}}} \right), \  m \in [K_{\mathrm{R}}-\tilde{r} +1 ,K_{\mathrm{R}}-1]_{\mathbb{Z}}.
\end{cases}
\end{equation}
%%%
Showing that $c'_m$ is non-increasing in $m$ is carried out through the  two following steps:
%%%
\begin{enumerate}
	\item We show that $d_m$ and $e_m$ are both non-increasing sequences in $m$ .
	This guarantees that  $c'_m$ is non-increasing over both the intervals $[1,K_{\mathrm{R}}-\tilde{r}-1 ]_{\mathbb{Z}}$ and  $[K_{\mathrm{R}}-\tilde{r}+1,K_{\mathrm{R}}-1]_{\mathbb{Z}}$.
	\item We show that $c'_{K_{\mathrm{R}}-\tilde{r}} \leq d_{K_{\mathrm{R}}-\tilde{r}-1}  $ and
	$c'_{K_{\mathrm{R}}-\tilde{r}} \geq e_{K_{\mathrm{R}}-\tilde{r}+1}  $. This guarantees that
	$c'_m$ is non-increasing over the entire interval  $[1,K_{\mathrm{R}}-1]_{\mathbb{Z}}$.
\end{enumerate}
%%%%

%%%
\emph{Proof of Point 1):} First, let us consider $d_m$.
This can be rewritten as:
\begin{equation}
\label{eq:dm}
d_m = \frac{(K_{\mathrm{R}}-m+1)(r-1)}{m(m+r-1)} + \frac{(K_{\mathrm{R}}-m)(1-r)}{(m+1)(m+r)}.
\end{equation}
For $r=1$, we have $d_m = 0$ for all $m \in [1, K_{\mathrm{R}}-1]_{\mathbb{Z}}$.
%%%
Hence, we consider $r \geq 1$.
%%%
From \eqref{eq:dm}, and after some rearrangements, the inequality $d_m \geq d_{m+1}$ which we wish to prove
is equivalently written as
\begin{equation}
\label{eq:dm_2}
\frac{K_{\mathrm{R}}-m+1}{m(m+r-1)} - \frac{K_{\mathrm{R}}-m}{(m+1)(m+r)} \geq
\frac{K_{\mathrm{R}}-m}{(m+1)(m+r)} - \frac{K_{\mathrm{R}}-m-1}{(m+2)(m+r+1)}.
\end{equation}
Using the following notation $A=K_{\mathrm{R}}-m$, $B=m+1$ and $C=m+r$, \eqref{eq:dm_2} is rewritten as
\begin{equation}
\label{eq:dm_3}
\frac{A+1}{(B-1)(C-1)} - \frac{A}{BC} \geq
\frac{A}{BC} - \frac{A-1}{(B+1)(C+1)}.
\end{equation}
After further rearranging and simplifying, \eqref{eq:dm_3} becomes
\begin{equation} \label{eq:dm_4}
ABC+B^2C+BC^2 \geq A-AB^2 - AC^2.
\end{equation}
Since $A \geq 1, B \geq 2$ and $C \geq 2$, \eqref{eq:dm_4} always holds
and hence $d_m$ is non-increasing in $m$.
%%
	
%%%
Next, we consider $e_m$. This can be rewritten as:
%%%
\begin{equation}
\label{eq:em}
e_m = \frac{K_{\mathrm{R}}+1}{m} + \frac{K_{\mathrm{R}}}{m(m+1)} - \frac{1}{m+1} -\frac{r}{m} -1
\end{equation}
From \eqref{eq:em}, it follows that $e_m \geq e_{m+1}$ is implied by
\begin{equation}
\label{eq:em2}
\frac{K_{\mathrm{R}}+1}{m} + \frac{K_{\mathrm{R}}}{m(m+1)} - \frac{1}{m+1} -\frac{r}{m} \geq  \frac{K_{\mathrm{R}}+1}{m+1} + \frac{K_{\mathrm{R}}}{(m+1)(m+2)} - \frac{1}{m+2} -\frac{r}{m+1}.
\end{equation}
After some rearrangements, the inequality in \eqref{eq:em2} becomes
$( K_{\mathrm{R}} + 1 -r)(m+2) + 2K_{\mathrm{R}} - m \geq 0$,
which holds as $m \geq 1$ and $K_{\mathrm{R}} \geq r$.
Hence, $e_m$ is a non-increasing in $m$ and this part is complete.

%%%
\emph{Proof of Point 2):}
%%%
In order to show that $c'_{K_{\mathrm{R}}-\tilde{r}} \leq d_{K_{\mathrm{R}}-\tilde{r}-1} $,
we only need to observe the following:
\begin{equation}
\nonumber
\begin{split}
	c'_{K_{\mathrm{R}}-\tilde{r}} & =\left( \frac{K_{\mathrm{R}}+1}{K_{\mathrm{R}}-\tilde{r}} - \frac{K_{\mathrm{R}}+\tilde{r} + \epsilon}{ K_{\mathrm{R}}+\epsilon -1 } \right)  + \frac{\tilde{r}}{K_{\mathrm{R}}-\tilde{r}}  \left( \frac{1}{K_{\mathrm{R}}-\tilde{r}+1} - \frac{\tilde{r} + \epsilon}{K_{\mathrm{R}}} \right)  \\
	& \leq  \left( \frac{K_{\mathrm{R}}+1}{K_{\mathrm{R}}-\tilde{r}} - \frac{K_{\mathrm{R}}+\tilde{r} + \epsilon}{ K_{\mathrm{R}}+\epsilon -1 } \right)  + \frac{\tilde{r}}{K_{\mathrm{R}}-\tilde{r}}  \left( \frac{1}{K_{\mathrm{R}}-\tilde{r}+1} - \frac{\tilde{r} + \epsilon}{K_{\mathrm{R}}+\epsilon} \right) \\
&  = d_{K_{\mathrm{R}}-\tilde{r}} \leq  d_{K_{\mathrm{R}}-\tilde{r}-1}.
	\end{split}
	\end{equation}
%%

%%%
Next, we focus on showing that $c'_{K_{\mathrm{R}}-\tilde{r}} \geq e_{K_{\mathrm{R}}-\tilde{r}+1} $.
We observe that $c'_{K_{\mathrm{R}}-\tilde{r}}$ can be expressed as:
\begin{equation} \label{eq:c_kr_ri}
	\begin{split}
	c'_{K_{\mathrm{R}}-\tilde{r}} & =\left( \frac{K_{\mathrm{R}}+1}{K_{\mathrm{R}}-\tilde{r}} - \frac{K_{\mathrm{R}}+\tilde{r} + \epsilon}{ K_{\mathrm{R}}+\epsilon -1 } \right)  + \frac{\tilde{r}}{K_{\mathrm{R}}-\tilde{r}} \left( \frac{1}{K_{\mathrm{R}}-\tilde{r}+1} - \frac{\tilde{r} + \epsilon}{K_{\mathrm{R}}} \right)  \\
	& =  \left( \frac{\tilde{r}+1}{K_{\mathrm{R}}-\tilde{r}} - \frac{\tilde{r} +1}{ K_{\mathrm{R}}+\epsilon -1 } \right)  + \frac{\tilde{r}}{K_{\mathrm{R}}-\tilde{r}} \left( \frac{1}{K_{\mathrm{R}}-\tilde{r}+1} - \frac{\tilde{r} + \epsilon}{K_{\mathrm{R}}} \right).
	\end{split}
\end{equation}
On the other hand, $e_{K_{\mathrm{R}}-\tilde{r}+1}$ is given by:
	\begin{equation} \label{eq:e_kr_ri}
	\begin{split}
	e_{K_{\mathrm{R}}-\tilde{r}+1} & =\left( \frac{K_{\mathrm{R}}+1}{K_{\mathrm{R}}-\tilde{r} + 1} - \frac{K_{\mathrm{R}}+\tilde{r} + \epsilon}{ K_{\mathrm{R}}} \right)  + \frac{\tilde{r}-1}{K_{\mathrm{R}}-\tilde{r}+1}  \left( \frac{1}{K_{\mathrm{R}}-\tilde{r}+2} - \frac{\tilde{r} + \epsilon}{K_{\mathrm{R}}} \right)  \\
	& =  \left( \frac{\tilde{r}}{K_{\mathrm{R}}-\tilde{r}+1} - \frac{\tilde{r} +\epsilon}{ K_{\mathrm{R}}} \right)  + \frac{\tilde{r}-1}{K_{\mathrm{R}}-\tilde{r}+1}  \left( \frac{1}{K_{\mathrm{R}}-\tilde{r}+2} - \frac{\tilde{r} + \epsilon}{K_{\mathrm{R}}} \right).
	\end{split}
	\end{equation}
By taking the difference of \eqref{eq:c_kr_ri} and \eqref{eq:e_kr_ri}, we obtain
%%%
\begin{multline}
\label{eq:ineq_c_prime_minus_e}
c'_{K_{\mathrm{R}}-\tilde{r}} - e_{K_{\mathrm{R}}-\tilde{r}+1}=\frac{K_{\mathrm{R}}+1-\tilde{r}-\epsilon}{(K_{\mathrm{R}}-\tilde{r})(K_{\mathrm{R}}-\tilde{r}+1)} + \frac{(\epsilon-1)(K_{\mathrm{R}}+\tilde{r}+\epsilon)}{K_{\mathrm{R}}(K_{\mathrm{R}}+\epsilon-1)}  \\ + \frac{K_{\mathrm{R}}+\tilde{r}}{(K_{\mathrm{R}}-\tilde{r})(K_{\mathrm{R}}-\tilde{r}+1)(K_{\mathrm{R}}-\tilde{r}+2)}.
\end{multline}
After rearranging the terms in \eqref{eq:ineq_c_prime_minus_e}, it follows that $c'_{K_{\mathrm{R}}-\tilde{r}} - e_{K_{\mathrm{R}}-\tilde{r}+1} \geq 0$ is implied by the inequality
\begin{multline}
\label{eq:ineq_parabolas}
\underbrace{K_{\mathrm{R}} (K_{\mathrm{R}} + \epsilon -1) (K_{\mathrm{R}} +1-\tilde{r}-\epsilon) (K_{\mathrm{R}} -\tilde{r}+2)}_{l_1(\epsilon)}
+ \underbrace{K_{\mathrm{R}} (K_{\mathrm{R}} + \epsilon -1) (K_{\mathrm{R}} + \tilde{r} )}_{l_2(\epsilon)} \\
+  \underbrace{(\epsilon - 1 )(K_{\mathrm{R}} +  \epsilon +\tilde{r}) (K_{\mathrm{R}} - \tilde{r} ) (K_{\mathrm{R}} - \tilde{r} + 1) (K_{\mathrm{R}} - \tilde{r} + 2)}_{l_3(\epsilon)} \geq 0.
\end{multline}
%%%
We denote the left-hand side of (\ref{eq:ineq_parabolas}) by $l(\epsilon) = l_1(\epsilon) +l_2(\epsilon) + l_3(\epsilon) $.
It is readily seen that $l_1(\epsilon)$ and $l_3(\epsilon)$ are second degree polynomials in the variable $\epsilon$
(i.e.  parabolas).
We consider the the three functions separately in order to derive a lower bound on $l(\epsilon)$.
\begin{itemize}
\item $l_1(\epsilon)$: It can be easily verified that $l_1(\epsilon)$ is concave with a maximum value at $\epsilon^*=\frac{2-\tilde{r}}{2}$.
%%%
Hence, $\epsilon^* \leq 0$ for $\tilde{r} \geq 2$ and $\epsilon^* = 1/2$ for $\tilde{r} = 1$.
%%%
As a concave parabola is decreasing for $\epsilon \geq \epsilon^*$ and symmetric with respect to the maximum, it follows that for $\epsilon \in [0,1)$, we have
\begin{equation} \label{eq:ineq_l1}
l_1(\epsilon) \geq l_1(1) = K_{\mathrm{R}}^2 (K_{\mathrm{R}} - \tilde{r}) (K_{\mathrm{R}} - \tilde{r} + 2).
\end{equation}
\item $l_2(\epsilon)$:
It is readily seen that for $\epsilon \in [0,1)$, the following holds
\begin{equation} \label{eq:ineq_l2}
l_2(\epsilon) \geq l_2(0) =  K_{\mathrm{R}} (K_{\mathrm{R}} - 1) (K_{\mathrm{R}} + \tilde{r}).
\end{equation}
%%
%%%
\item $l_3(\epsilon)$: This is a convex with a minimum value at $\epsilon^*=\frac{-K_{\mathrm{R}}-\tilde{r}+1}{2} < 0$.
Hence, for $\epsilon \in [0,1)$, we have
\begin{equation} \label{eq:ineq_l3}
l_3(\epsilon) \geq l_3(0) = - (K_{\mathrm{R}} + \tilde{r}) (K_{\mathrm{R}} - \tilde{r}) (K_{\mathrm{R}} - \tilde{r} + 1) (K_{\mathrm{R}} - \tilde{r} + 2).
\end{equation}
\end{itemize}
By summing over the  lower bounds in \eqref{eq:ineq_l1}, \eqref{eq:ineq_l2} and \eqref{eq:ineq_l3},
it follows that for $\epsilon \in [0,1)$, we have:
\begin{equation} \label{eq:inequality_lb}
l(\epsilon) \geq  K_{\mathrm{R}} (K_{\mathrm{R}} - 1) (K_{\mathrm{R}} + \tilde{r})  + (K_{\mathrm{R}} - \tilde{r}) (K_{\mathrm{R}} - \tilde{r} + 2) (\tilde{r}^2 - \tilde{r} - K_{\mathrm{R}}).
\end{equation}
Next, we express the right-hand side of the \eqref{eq:inequality_lb} as a function of  $K_{\mathrm{R}}$:
%%%
\begin{equation}
\nonumber
g(K_{\mathrm{R}}) = K_{\mathrm{R}} (K_{\mathrm{R}} - 1) (K_{\mathrm{R}} + \tilde{r})  + (K_{\mathrm{R}} - \tilde{r}) (K_{\mathrm{R}} - \tilde{r} + 2) (\tilde{r}^2 - \tilde{r} - K_{\mathrm{R}}) = a K_{\mathrm{R}}^2+bK_{\mathrm{R}}+c
\end{equation}
%%%
where $a=\tilde{r}^2 +2 \tilde{r} -3$ and $b = -\tilde{r} (2\tilde{r}^2 - 3 \tilde{r} +1)$.
%%%
Finally, to show that $l(\epsilon)  \geq 0$, it is sufficient to show $g(K_{\mathrm{R}}) \geq 0$ for all $K_{\mathrm{R}} \geq \tilde{r}$.
%%%
To this end, we observe that $g(K_{\mathrm{R}}) = 0$ for $\tilde{r}=1$, while $g(K_{\mathrm{R}})$ is a convex parabola with a minimum value at
$\frac{\tilde{r} (\tilde{r} - 1/2)}{\tilde{r} + 3} \leq \tilde{r}$ for $\tilde{r} > 1$.
In latter case, $g(K_{\mathrm{R}})$ is increasing for $K_{\mathrm{R}} \geq \tilde{r}$.
As $g(\tilde{r}) \geq 0$, it follows that $g(K_{\mathrm{R}}) \geq 0$ for all $K_{\mathrm{R}} \geq \tilde{r}$. This concludes the proof.
%%%%%%%%%%%%%%%%%%%
\section*{Acknowledgment}
%%%
The authors are grateful to Prof. Iosif Pinelis for proving Lemma \ref{lemma:inequality}.
%%%%%%%%%%%%%%%%%%%
%%%%%%%%%%%%%%%%%%%
\bibliographystyle{IEEEtran}
\bibliography{References}

\end{document}